\documentclass[letterpaper,final,twocolumn]{IEEEtran}  % Comment 
\usepackage{amsmath}
\usepackage{amssymb,amsthm,graphicx,verbatim,enumitem,cite}
\usepackage{color}

\usepackage{algorithm,algorithmic}
\usepackage{subfigure}
\usepackage[colorlinks = true, linkcolor = blue, citecolor = 
blue]{hyperref}

\newcommand{\mycomment}[1]{\hfill \COMMENT{\texttt{#1}}}

\newtheorem{theorem}{Theorem}[section]
\newtheorem{proposition}[theorem]{Proposition}
\newtheorem{lemma}[theorem]{Lemma}
\newtheorem{corollary}[theorem]{Corollary}
\newtheorem{remark}[theorem]{Remark}

\newcommand{\real}{\ensuremath{\mathbb{R}}}
\newcommand{\realpositive}{\ensuremath{\mathbb{R}_{>0}}}
\newcommand{\realnonnegative}{\ensuremath{\mathbb{R}_{\ge 0}}}
\newcommand{\integerspositive}{{\mathbb{N}}}
\newcommand{\integersnonnegative}{{\mathbb{N}}_0}
\newcommand{\longthmtitle}[1]{\mbox{}\textit{(#1).}}

\newcommand{\intersect}{\ensuremath{\operatorname{\cap}}}

\DeclareMathOperator*{\argmin}{argmin}

\newcommand{\Bc}{\mathcal{B}}

\newcommand{\Sc}{\mathcal{S}}
\newcommand{\Vc}{\mathcal{V}}
\newcommand{\Xc}{\mathcal{X}}

\newcommand{\Enorm}[1]{\|#1\|_{2}}
\newcommand{\Infnorm}[1]{\|#1\|_{\infty}}
\newcommand{\uline}[1]{\underline{#1}}
\newcommand{\volume}[1]{\operatorname{vol}(#1)}
\newcommand{\trace}[1]{\operatorname{tr}(#1)}

\newcommand{\triggerCh}{h_{\operatorname{ch}}}

\newcommand{\triggerChbar}{\bar{h}_{\operatorname{ch}}}
\newcommand{\rhofun}[2]{\rho_{#1}(#2)}
\newcommand{\bth}{\bar{\theta}}

\newcommand{\setdef}[2]{\left\{#1 \, :\, #2\right\}}

\newcommand{\identity}[1]{\mathrm{I}_{#1}}
\newcommand{\zeros}[1]{\mathrm{0}_{#1}}
\newcommand{\realpart}[1]{\mathrm{Re}({#1})}

\newcommand{\oprocendsymbol}{\hbox{$\bullet$}}
\newcommand{\oprocend}{\relax\ifmmode\else\unskip\hfill\fi\oprocendsymbol}

\graphicspath{{figures/}}

\begin{document}

% Event-Triggered Control with Bounded Data Rate and under Disturbance

\title{Event-Triggered Stabilization of Linear Systems Under Bounded
  Bit Rates\thanks{A preliminary version of this work appeared at the
    2014 IEEE Conference on Decision and Control.}}

\author{Pavankumar Tallapragada \qquad Jorge Cort{\'e}s
  \thanks{Pavankumar Tallapragada and Jorge Cort{\'e}s are with the
    Department of Mechanical and Aerospace Engineering, University of
    California, San Diego {\tt\small
      \{ptallapragada,cortes\}@ucsd.edu}}%
}

\maketitle

\begin{abstract}
  This paper addresses the problem of exponential practical
  stabilization of linear time-invariant systems with disturbances
  using event-triggered control and bounded communication bit rate.
  We consider both the case of instantaneous communication with finite
  precision data at each transmission and the case of
  non-instantaneous communication with bounded communication rate.
  Given a prescribed rate of convergence, the proposed event-triggered
  control implementations opportunistically determine the transmission
  instants and the finite precision data to be transmitted on each
  transmission. We show that our design exponentially practically
  stabilizes the origin while guaranteeing a uniform positive lower
  bound on the inter-transmission and inter-reception times, ensuring
  that the number of bits transmitted on each transmission is upper
  bounded uniformly in time, and allowing for the possibility of
  transmitting fewer bits at any given time if more bits than
  prescribed were transmitted earlier.  We also characterize the
  necessary and sufficient average data rate for exponential practical
  stabilization. Several simulations illustrate the results.
\end{abstract}

\vspace*{-1.5ex}
\section{Introduction}

The digital nature of communication in networked control systems
naturally induces sampling and quantization of signals.  The
increasing ubiquity of these systems, particularly in
resource-constrained domains where communication channels have low,
time-varying, and possibly unreliable channel capacity, has brought to
the forefront the need for integrated and systematic design
methodologies that go beyond adhoc approaches.  This paper is a
contribution to the modern body of research that seeks to
fundamentally address the problem of control under constrained
resources. Specifically, we seek to combine the strengths of
event-triggered control and information theory to efficiently
stabilize linear time-invariant systems under communication
constraints.

\emph{Literature Review:} The need for % systems
integration % and the importance of bridging the gap
between computing, communication, and control in the study of
cyberphysical systems cannot be
overemphasized~\cite{KDK-PRK:12,JS-XK-GK-NK-PA-VG-BG-JB-SW:12}.  The
present work builds on two areas of research that address the
stabilization of control systems under limited information from
different and complementary perspectives.  In the
information-theoretic approach to control under communication
constraints, the focus is on determining sufficient and necessary
conditions on the \emph{bit data rates} (i.e., the number of bits
transmitted over possibly multiple transmissions during an arbitrary
time interval) that guarantee stabilization under varying assumptions
on the communication channels.  The
works~\cite{GNN-FF-SZ-RJE:07,MF-PM:14} provide comprehensive accounts
of this by now vast literature, and we highlight next a few references
most relevant to the discussion here.  Early data rate results
appeared in~\cite{GNN-RJE:00,GNN-RJE:04,ST-SM:04}, which employ the
idea of countering the information generated (the growth in the
uncertainty of the system state) with a sufficiently high data rate of
the encoded feedback.  This approach has been successful in providing
tight necessary and sufficient conditions on the data rate of the
encoded feedback for asymptotic stabilization in the discrete-time
setting. Similar ideas have been used to provide data rate theorems
also for stochastic rate channels~\cite{NM-MD-NE:06} and extended to
vector systems and time-varying feedback
channels~\cite{PM-MF-SD-GNN:09} and Markov feedback
channels~\cite{PM-LC-MF:13}. In the continuous-time setting, the
problem has been mainly studied under either periodic sampling or
aperiodic sampling with known upper and lower bounds on the sampling
period for single input systems~\cite{LL-JB:04,LL-JB:07}, nonlinear
feedforward systems~\cite{CDP:05}, and switched linear
systems~\cite{DL:14}. In this context, it is not known if and how a
best sampling period may be designed or if state-based aperiodic
sampling can provide any advantage in efficiency and performance. With
a few exceptions, see e.g.,~\cite{DL:14}, the works above do not
characterize the convergence rates or explore the problem of
guaranteeing a desired performance.

Event-triggered control, instead, trades computation and
decision-making for less communication, sensing, or actuation effort,
while guaranteeing a desired level of performance. This literature,
see e.g.~\cite{PT:07,XW-MDL:11,WPMHH-KHJ-PT:12} and references
therein, exploits the tolerance to measurement errors to design
goal-driven state-based aperiodic sampling for the efficient use of
the system resources. The main focus of this body of work is on
minimizing the number of updates while guaranteeing the feasibility of
the resulting real-time implementation.  When interpreted in terms of
communication, this results in a paradigm where one seeks to minimize
the number of transmissions while largely ignoring the quantization
aspect and allowing the data at each transmission to be of infinite
precision.  Among the few exceptions, we mention event-triggered
schemes with static logarithmic quantization~\cite{PT-NC:12,EG-PJA:13}
and dynamic quantization~\cite{DL-JL:10, LL-XW-MDL:12, LL-XW-MDL:12a,
  LL-BH-MDL:12, YS-XW:14}.  In~\cite{PT-NC:12}, events are defined as
the system state crossing static quantization cells and communication
is assumed to be instantaneous and there are no
disturbances. \cite{EG-PJA:13} considers the problem with modeling
errors and communication delays. Both these papers do not explicitly
study the notion of \emph{communication bit rate} (i.e., the number of
bits per transmission).
In~\cite{DL-JL:10,LL-XW-MDL:12,LL-XW-MDL:12a,LL-BH-MDL:12}, the events
are defined as the infinity norm of the encoding error crossing a
fixed or piecewise-constant threshold. \cite{DL-JL:10} considers
instantaneous communication and external disturbances, although the
use of a fixed threshold in the event-triggering condition results in
practical stability even under no disturbance. In addition, if the
channel imposes a bound on the communication bit rate, then it also
affects the ultimate bound on the state. \cite{LL-XW-MDL:12} addresses
the problem for nonlinear systems and with communication delays,
while~\cite{LL-XW-MDL:12a, LL-BH-MDL:12, YS-XW:14} extend these
results to the case with external disturbance.  All these works
guarantee a positive lower bound on the inter-transmission times,
while~\cite{DL-JL:10, LL-XW-MDL:12, LL-XW-MDL:12a, LL-BH-MDL:12,
  YS-XW:14} also provide a uniform bound on the communication bit
rate. However, these references do not address the inverse problem of
triggering and quantization given a limit on the communication bit
rate imposed by the channel.  While guarantees on the uniform
boundedness of the communication bit rate are useful, they do not
characterize either necessary or sufficient conditions on the required
data rates, i.e., the number of bits averaged over a finite or
infinite time horizon.  In fact, this is a shortcoming of the
event-triggered control literature as a whole, where the availability
of such analytical results would help in the design of networked
control systems. Finally, the common underlying approach in the
event-triggered literature is based on the notion of
input-to-state-stability with respect to measurement errors for both
event-triggering and quantization. This is in contrast with the
information-theoretic data rate approach to quantization and encoding
adopted here.

\emph{Statement of Contributions:} This paper designs event-triggered
controllers for linear-time invariant systems under bounded
communication bit rate.  We focus on the control goal of exponential
practical stabilization, in the presence of disturbance and with a
prescribed rate of convergence. The first contribution is the
identification of a necessary condition on the average data rate
required for all solutions of a linear-time invariant system to
exponentially converge with a prescribed convergence rate. Our second
set of contributions pertain to the design of event-triggered
controllers that guarantee exponential convergence with a desired
performance by adjusting the communication rate in accordance with
state information in an opportunistic fashion.  We consider
increasingly realistic scenarios, ranging from instantaneous
transmissions with arbitrary, but finite communication rate, through
instantaneous transmissions with uniformly bounded communication rate,
to finally non-instantaneous transmissions with arbitrary bounded
communication rate imposed by the channel. In all cases, our design
guarantees the existence of a uniform positive lower bound on
inter-transmission and inter-reception times, and ensures that the
number of bits transmitted at each transmission is upper bounded.  An
overarching contribution of the paper is the introduction of the
information-theoretic data rate approach to quantization and encoding
to complement event-triggering for data rate limited feedback control.
From an event-triggered control perspective, our key contribution is
going beyond the paradigm of infinite precision at each transmission
and adopting the information-theoretic approach to quantization,
encoding, and triggering.  This allows us to characterize necessary
and sufficient data rates averaged over time, and quantify the
capability to transmit fewer bits if more bits than prescribed were
transmitted earlier.  From an information-theoretic perspective, our
key contribution is the efficient use of the communication resources
by exploiting state-based opportunistic sampling. This allows us to
tune the operation of the control system to the desired level of
performance and guarantee a desired convergence rate.  In order to
communicate the main ideas effectively, we use a simple encoding
scheme and assume that the encoder and the decoder know the
communication delays in the case of non-instantaneous
communication. These are aspects that may be improved upon within the
framework of the paper. Finally, we believe the approach laid out here
opens up numerous avenues for further research at the intersection of
information theory, control, and stabilization.

\emph{Organization:} Section~\ref{sec:prob_stat} formally states the
asymptotic stabilization problem under event-triggered control and
finite communication bit rate. Section~\ref{sec:necess_cond}
identifies a necessary condition on the average data rate required for
all solutions to asymptotically converge with a prescribed convergence
rate. Sections~\ref{sec:ET_bnd_dat_rate}
and~\ref{sec:ET_bnd_dat_rate-non-inst} present our event-triggered
control design with bounded communication rate under instantaneous and
non-instantaneous communication, respectively.  Section~\ref{sec:sim}
presents simulation results. Finally, Section~\ref{sec:conc} gathers
our conclusions and ideas for future work. Proofs of certain auxiliary
lemmas are presented in the appendix for smoother readability.

\emph{Notation:} We let $\real$, $\realnonnegative$,
$\integerspositive$, and $\integersnonnegative$ denote the set of
real, nonnegative real, positive integer, and nonnegative integer
numbers, respectively.  We let $\identity{n}$ and $\zeros{n} \in
\real^{n \times n}$ denote the identity and zero matrix, respectively,
of dimension $n$. For a matrix $A$, let $\sigma_A$ denote the spectrum
of the matrix $A$ and $\realpart{\sigma_A}$ denote the set of real
parts of the eigenvalues of $A$. For a symmetric matrix $A \in
\real^{n\times n}$, we let $\lambda_m(A)$ and $\lambda_M(A)$ denote
its smallest and largest eigenvalues, respectively. For a symmetric
positive definite matrix $P \in \real^{n\times n}$ and all $ x \in
\real^n$,
\begin{align}\label{eq:quadratic-form-ineq}
  \sqrt{\lambda_m(P)} \Enorm{x} \leq \sqrt{x^T P x} \leq
  \sqrt{\lambda_M(P)} \Enorm{x} .
\end{align}
Given $A_1, A_2 \in \real^{n \times n}$, $A_1 \prec A_2$ denotes that
$A_1 - A_2$ is negative definite. Similarly, the symbols $\preceq$,
$\succ$ and $\succeq$ stand for negative semi-definiteness, positive
definiteness and positive semi-definiteness, respectively. We denote
by $\Enorm{.}$ and $\Infnorm{.}$ the Euclidean and infinity norm of a
vector, respectively, or the corresponding induced norm of a
matrix. For $A \in \real^{n \times m}$, we let $A^+$ denote the
pseudoinverse. For any matrix norm $\|.\|$, note that $\|e^{A \tau}\| \leq
e^{\|A\| \tau}$. For a function $f: \real \mapsto \real^{n}$
and any $t \in \real$, we let $f(t^-)$ denote the limit from the left,
$\displaystyle \lim_{s \uparrow t} f(s)$.

\section{Problem Statement}\label{sec:prob_stat}

Consider a % plant whose dynamics is given by a
linear time-invariant
control system,
\begin{align}\label{eqn:plant_dyn}
  \dot{x}(t) = A x(t) + B u(t) + v(t),
\end{align}
where $x \in \real^n$ denotes the state of the plant, $u \in \real^m$
is the control input and $v \in \real^n$ is an unknown
disturbance. Here, $A \in \real^{n \times n}$ and $B \in \real^{n
  \times m}$ are the system matrices.  We assume that the pair $(A,B)$
is stabilizable, i.e., there exists a control gain matrix $K \in
\real^{m \times n}$ such that the matrix $\bar{A} = A+BK$ is Hurwitz,
and that the disturbance is uniformly bounded by a known constant,
i.e.,
\begin{equation}\label{eqn:disturb_bound}
  \Enorm{v(t)} \leq \nu, \ \forall t \in [0, \infty) .
\end{equation}
Under these assumptions, $u(t) = K x(t)$ renders the origin
of~\eqref{eqn:plant_dyn} globally exponentially practically stable.

The plant is equipped with a sensor and an actuator, which are not
co-located. We assume that the sensor can measure the state exactly,
and that the actuator can exert the input to the plant with infinite
precision. However, the sensor has the ability to transmit state
information to the controller at the actuator only at discrete time
instants of its choice and using only a finite number of bits. In this
sense, we refer to the sensor as the encoder and the actuator as the
decoder. We let $\{ t_k \}_{k \in \integerspositive}$ be the sequence
of \emph{transmission (or encoding) times} at which the sensor decides
to sample, encode, and transmit the plant state.  We denote by $n p_k$
the number of bits used to encode the plant state at the transmission
time~$t_k$.  The process of encoding, transmission by the sensor,
reception of a complete packet of encoded data at the controller, and
decoding may take non-zero time.  We let $\{r_k\}_{k \in
  \integerspositive}$ be the sequence of \textit{reception (or update)
  times} at which the decoder receives a complete packet of data,
decodes it, and updates the controller state. Therefore, $r_k \ge
t_k$. The $k^\text{th}$ communication time $\Delta_k \triangleq r_k -
t_k$ is then a function of $t_k$ and the packet size (of $n p_k$ bits)
represented by $p_k$,
\begin{equation*}
  \Delta_k = r_k - t_k \triangleq \Delta( t_k, p_k ) .
\end{equation*}
In general, the time $\Delta_k$ could include communication time,
computation time and other delays. We refer to the case $\Delta \equiv
0$ by instantaneous communication.  To keep things simple, we assume
the encoder and the decoder have synchronized clocks and synchronously
update their states at update times $\{r_k\}_{k \in
  \integerspositive}$. The latter assumption is justified in
situations where $t \mapsto \Delta(t,p)$ is independent of $t$ or
where the encoder and decoder send short synchronization signals to
indicate the start of encoding and the end of decoding,
respectively. % We discuss this further later in
% Section~\ref{sec:ET_bnd_dat_rate-non-inst}.

We use dynamic quantization for finite-bit transmissions from the
encoder to the decoder. In dynamic quantization, there are two
distinct phases: the zoom-out stage, during which no control is
applied while the quantization domain is expanded until it captures
the system state at time $r_0 = t_0 \in \realnonnegative$; and the
zoom-in stage, during which the encoded feedback is used to
asymptotically stabilize the system. A detailed description of the
zoom-out stage can be found in the literature,
e.g.,~\cite{DL:03}. Here, we focus exclusively on the zoom-in stage,
i.e., for $t \geq t_0$ for which we use a hybrid dynamic controller.
We assume that both the encoder and the decoder have perfect knowledge
of the plant system matrices.  The state of the encoder/decoder is
composed of the controller state $\hat{x} \in \real^n$ and an upper
bound $d_e \in \realnonnegative$ on the norm of the encoding error
$x_e \triangleq x - \hat{x}$.  Thus, the actual input to the plant is
given by $u(t) = K \hat{x}(t)$. During inter-update times, the state
of the dynamic controller evolves as
\begin{subequations}\label{eqn:x_hat}
  \begin{align}\label{eqn:xhat_evolve}
    \dot{\hat{x}}(t) = A \hat{x}(t) + B u(t) = \bar{A} \hat{x}(t),
    \quad t \in [r_k, r_{k+1}) .
  \end{align}
  Let the encoding and decoding functions at $k^{\text{th}}$ iteration
  be represented by $q_{E, k}: \mathbb{R}^n \times \mathbb{R}^n
  \mapsto G_k$ and $q_{D, k}: G_k \times \mathbb{R}^n \mapsto
  \mathbb{R}^n$, respectively, where $G_k$ is a finite set of
  $2^{n p_k}$ symbols. At $t_k$, the encoder encodes the plant state as
  $z_{E,k} \triangleq q_{E,k}(x(t_k), \hat{x}(t_k^-))$, where
  $\hat{x}(t_k^-)$ is the controller state just prior to the encoding
  time $t_k$, and sends it to the controller. This signal is decoded
  as $z_{D,k} \triangleq q_{D,k}(z_{E,k}, \hat{x}(t_k^-))$ by the
  decoder at time $r_k$. Then at the update time $r_k$, the sensor 
  and the controller update $\hat{x}$ using the jump map,
  \begin{align}\label{eqn:x_hat_jmp_nominal}
    \hat{x}(r_k) & = e^{\bar{A} \Delta_k} \hat{x}(t_k^-) + e^{A
      \Delta_k} (z_{D,k} - \hat{x}(t_k^-)) \notag
    \\
    & \triangleq q_k(x(t_k), \hat{x}(t_k^-)) .
  \end{align}
\end{subequations}
We use the shorthand notation $q_k: \mathbb{R}^n \times \mathbb{R}^n
\mapsto \mathbb{R}^n$ to represent the quantization that occurs as a
result of the finite-bit coding. We allow the quantization domain, the
number of bits and the resulting quantizer, $q_k$, at each
transmission instant $t_k \in \realnonnegative$ to be variable. Note
that the evaluation of the map $q_k$ is inherently from the encoder's
perspective because it depends on the plant state $x(t_k)$, which is
unknown to the decoder. Also, while the encoder could store
$\hat{x}(t_k^-)$, the decoder has to infer its value if
$\Delta_k>0$. We detail the specifics of the decoder's procedure to
implement~\eqref{eqn:x_hat_jmp_nominal} when communication is not
instantaneous later.

The evolution of the plant state $x$ and the encoding error~$x_e$ on
the time interval $[r_k, r_{k+1})$ can be written as
\begin{subequations}\label{eq:dynamics-x-xe}
  \begin{align}
    \dot{x}(t) &= \bar{A} x(t) - BK x_e(t) + v(t), 
    \label{eqn:plant_dyn_enc_err}
    \\
    \dot{x}_e(t) &= A x_e(t) + v(t). \label{eqn:enc_err_dyn}
  \end{align}
\end{subequations}
Note that while the controller state $\hat{x}$ is known to both the
encoder and the decoder, the plant state (equivalently, the encoding
error $x_e$) is known only to the encoder. However, at $t_0$, if a
bound on $\Infnorm{x_e(t_0)}$ is available, then both the encoder and
the decoder can compute a bound $d_e(t)$ on $\Infnorm{x_e(t)}$ for any
$t \in \realnonnegative$, as we explain later.

Finally, in order to formalize the control goal, we select an
arbitrary symmetric positive definite matrix $Q \in \real^{n \times
  n}$. Because $\bar{A}$ is Hurwitz, there exists a symmetric positive
definite matrix $P$ that satisfies the Lyapunov equation
\begin{align}\label{eq:Lyap-eq}
  P \bar{A} + \bar{A}^T P = - Q.
\end{align}
Consider then the associated candidate Lyapunov function $x \mapsto
V(x) = x^T P x$. Given a desired ``control performance''
\begin{equation}\label{eqn:Vd_t}
  V_d(t) = ( V_d(t_0) - V_0) e^{ -\beta (t - t_0)} + V_0
\end{equation}
with $V_0 \geq 0$ (the steady state value of $V_d$) and $\beta > 0$
(rate of convergence) constants, the \emph{control objective} is as
follows: recursively determine the sequence of transmission times
$\{t_k\}_{k \in \integerspositive} \subset \realpositive$ and encoded
messages $\hat{x}(t_k)$ so that $V(x(t)) \leq V_d(t)$ holds for all $t
\geq t_0$, while also ensuring that the inter-transmission times $\{
t_k - t_{k-1} \}_{k \in \integerspositive}$ are uniformly lower
bounded by a positive quantity and that the number of bits transmitted
at any instant is uniformly upper bounded.  We structure our solution
to this problem in several stages. Section~\ref{sec:necess_cond}
presents a necessary condition on the average data rate required to
meet the control objective under the assumption of zero
disturbance. In Section~\ref{sec:ET_bnd_dat_rate} we address the
problem under instantaneous communication. Finally, we address the
problem in all its generality in
Section~\ref{sec:ET_bnd_dat_rate-non-inst}.

\section{Lower Bound on the Necessary Data
  Rate}\label{sec:necess_cond}

Here we seek to determine the amount of information, in terms of the
number of bits transmitted, necessary to meet the control goal stated
in Section~\ref{sec:prob_stat} for arbitrary initial conditions when
no disturbances are present and communication is instantaneous.  In
the presence of unknown disturbances and/or non-instantaneous
communication, the necessary data rate is at least as much as in the
case treated here, so the necessary condition also holds in those
cases. For convenience, let $\Bc(t,t_0)$ denote the number of bits
transmitted in the time interval $[t_0, t]$. We are also interested in
characterizing the data rate (i.e., the average number of bits
transmitted) asymptotically,
\begin{equation*}
  R_{\text{as}} \triangleq \lim_{t \rightarrow \infty} \frac{\Bc(t,t_0)}{t 
    - t_0}.
\end{equation*}

Since encoding is not exact, the decoder at the controller has
knowledge of the plant state only up to some set $\Sc(t) \subset
\real^n$, i.e., $x(t) \in \Sc(t)$. We refer to $\Sc(t)$ as the
\emph{state uncertainty set} at time $t$. Equivalently, the decoder
has knowledge of the encoding error $x_e(t)$ only up to some set $E(t)
\subset \real^n$, i.e., $x_e(t) \in E(t)$. Because $\hat{x}$ is known
to both the encoder and the decoder, $\Sc(t)$ is simply obtained as a
coordinate shift of the set $E(t)$,
\begin{equation*}
  \Sc(t) = \{ \xi \in \real^n : \xi = \hat{x}(t) + \xi_e, \ \xi _e \in
  E(t) \} .
\end{equation*}
Since $x_e(t_k) \in E(t_k)$ for each $k \in \integersnonnegative$,
then equation~\eqref{eqn:enc_err_dyn}, with $v(t) \equiv 0$, implies
that, for $ t \in [t_k, t_{k+1})$,
\begin{align}\label{eq:E}
  E(t) = \{ \xi \in \real^n : \xi = e^{A(t-t_k)} \xi_0, \ \xi_0 \in
  E(t_k) \},
\end{align}
where $E(t_0)$ is known to the encoder at the end of the zoom-out
stage of the dynamic quantization. If $A$ is not Hurwitz, then this
set grows with time unless some new information is communicated to the
controller.  To meet the specified control goal, the idea is to keep
the encoding error set $E(t)$ sufficiently small at all times by
having the sensor transmit information to the controller at the time
instants $t_k$.

\begin{remark}\longthmtitle{Reduction in the Bound on the Encoding 
    Error with Communication}\label{rem:E_tk_minmax_bnd}
  {\rm Suppose the sensor encodes the state $x(t_k)$ at $t_k$ using
    $n p_k$ bits by partitioning the set $E(t_k^-)$ (or equivalently
    $\Sc(t_k^-)$ ) into $2^{n p_k}$ subsets in a predetermined manner. The
    string of $n p_k$ bits informs the decoder the specific subset that
    $x(t_k)$ lies in. Further, suppose that $\hat{x}(t_k)$ is chosen
    as a nominal point of $\Sc(t_k)$ according to some predetermined
    rule. Then, note that there is some $x_e(t_k) \in E(t_k^-)$ such
    that, after performing the quantization,
    \begin{align*}
      \volume{E(t_k)} \geq \frac{\volume{E(t_k^-)}}{2^{n p_k}} ,
    \end{align*}
    where $\volume{S}$ denotes the volume of the set~$S$. The equality
    is achieved when the quantization (partitioning of the
    quantization domain) is uniform.} \oprocend
\end{remark}

The following result precisely characterizes the number of bits that
\textit{must} be transmitted to make it possible for the set $\Sc(t)$
(which has the same volume as $E(t)$) to be contained in
$\mathcal{V}_d(t) = \{ \xi \in \mathbb{R}^n : V(\xi) \leq V_d(t) \}$
as a means to ensure for \textit{every} solution satisfying $V(x(t_0))
\leq V_d(t_0)$ at time $t_0$ to also satisfy $x(t) \in
\mathcal{V}_d(t)$ for all $t \ge t_0$. Note that $\Vc_d(t)$ is a
sub-level set of the quadratic function $V(x) = x^T P x$. Thus,
$\Vc_d(t)$ is an $n$-dimensional ellipsoid, which by expressing as a
linear transformation of an $n$-sphere of radius $\sqrt{V_d(t)}$ gives
its volume to be
\begin{equation}\label{eqn:vol_Vc_d}
  \volume{\Vc_d(t)} = c_P ( V_d(t) )^{\frac{n}{2}}
\end{equation}
with
\begin{equation*}
  c_P \triangleq \sqrt{\det(P^{-1})} \frac{ \pi^{n/2} }{
    \Gamma(\frac{n}{2} + 1) } ,
\end{equation*}
where $\Gamma$ is the gamma function.  We are now ready to state the
result.

\begin{proposition}\longthmtitle{Necessary Number of Bits Transmitted
    and Asymptotic Data Rate}\label{prop:necs_rate} 
  Consider the system~\eqref{eqn:plant_dyn}, with $\min
  \realpart{\sigma_{A + \beta \identity{n}}} \geq 0$, $v(t) \equiv 0$
  and $V_0 = 0$, and under the feedback law $u(t) = K \hat{x}(t)$,
  where $t\mapsto \hat{x}(t)$ evolves according to~\eqref{eqn:x_hat}.
  A necessary condition for all solutions satisfying $V(x(t_0)) \leq
  V_d(t_0)$ at time $t_0$ to satisfy $V(x(t)) \leq V_d(t)$ for $t \ge
  t_0$ is
  \begin{multline}\label{eqn:Bc_t}
    \Bc(t,t_0) \geq \Big( \trace{A} + \frac{n \beta}{2} \Big) \log_2
    (e) (t-t_0)
    \\
    \quad + \log_2 \left( \frac{\volume{E(t_0)}}{c_P ( V_d(t_0)
        )^{\frac{n}{2}}} \right) .
  \end{multline}
  Consequently, $ R_{\text{as}} \geq \big( \trace{A} + \frac{n
    \beta}{2} \big) \log_2 (e) $.
\end{proposition}
\begin{proof}
  The main idea behind the proof is that in order for all solutions
  with initial conditions such that $V(x(t_0)) \leq V_d(t_0)$ to
  satisfy $V(x(t)) \leq V_d(t)$ for $t \ge t_0$ then it is necessary
  that the state uncertainty set $\Sc(t) \subseteq \Vc_d(t)$ at each
  time $t \geq t_0$. In particular, this implies that the volume of
  the set $\Sc(t)$ (or equivalently that of the coordinate-shifted
  $E(t)$) must be no greater than that of $\Vc_d(t)$, i.e., it is
  necessary that $\volume{E(t)} \leq \volume{\Vc_d(t)}$ for all $t
  \geq t_0$.

  Given a sequence of transmission times $\{t_k\}_{k \in
    \integerspositive}$, we deduce from~\eqref{eq:E} that for $t \in
  [t_k,t_{k+1})$,
  \begin{align*}
    \frac{\volume{E(t)}}{\volume{E(t_k)}} = \det \big( e^{A (t - t_k)}
    \big) = e^{\trace{A}(t-t_k)}.
  \end{align*}
  Further, if $\Bc(t,t_0)$ number of bits are transmitted in the time
  interval $[t_0, t]$, then as a consequence of
  Remark~\ref{rem:E_tk_minmax_bnd} it follows that there exists some
  $x(t_0)$ such that
  \begin{align}\label{eq:auxx}
    \volume{E(t)} \geq \frac{e^{\trace{A}(t-t_0)}
      \volume{E(t_0)}}{2^{\Bc(t,t_0)}}.
  \end{align}  
  Next, using $V_d(t) = V_d(t_0) e^{-\beta(t-t_0)}$ and
  \eqref{eqn:vol_Vc_d}, we deduce that
  \begin{align*}
    \volume{\mathcal{V}_d(t)} = c_P ( V_d(t_0) )^{\frac{n}{2}} e^{-
      \frac{ n \beta}{2} (t-t_0)}.
  \end{align*}
  Combining the observations in the beginning of the proof,
  and~\eqref{eq:auxx}, we require
  \begin{align*}
    2^{\Bc(t,t_0)} & \geq \frac{e^{\trace{A}(t-t_0)}
      \volume{E(t_0)}}{\volume{\mathcal{V}_d(t)}}
    \\
    & = \frac{e^{(\trace{A} + \frac{n \beta}{2} )(t-t_0)}
      \volume{E(t_0)}}{ c_P ( V_d(t_0) )^{\frac{n}{2}}} ,
  \end{align*}
  from which the result follows.
\end{proof}

If $\min \realpart{\sigma_{A + \beta \identity{n}}} < 0$ then the
eigen-subspace corresponding to the eigenvalues whose real part is
less than $\beta$ may be ignored without loss of generality. Thus, the
result is consistent with the well-known data-rate
theorem~\cite{GNN-FF-SZ-RJE:07, MF-PM:14}, which is obtained by
choosing $\beta = 0$.

There are a few observations of note regarding
Proposition~\ref{prop:necs_rate}. First, the condition is dependent on
the control goal but not on the control input itself. Since the result
only relies on comparing the volumes of the sets $\Sc(t)$ and
$\mathcal{V}_d(t)$, rather than on ensuring the stricter condition
$\Sc(t) \subseteq \mathcal{V}_d(t)$ for $t \geq t_0$, it remains to be
seen how a necessary or even a sufficient data rate condition would
depend on the control gain $K$ and the sequence of communication times
$\{t_k\}_{k \in \integersnonnegative}$. In general, a time-triggered
implementation with the given control goal and communication
constraints could be very conservative. This motivates our forthcoming
investigation of event-triggered designs. Furthermore, note that
Proposition~\ref{prop:necs_rate} is a necessary condition to meet the
control goal \emph{for every possible solution}. It is true that if
the decoder at the controller were deciding the transmission time
instants, then the condition $\Sc(t) \subset \mathcal{V}_d(t)$, $t
\geq t_0$, would have to be enforced (given that it has no access to
the actual plant state). However, when the encoder at the sensor is
deciding the transmission time instants, as in our case, then it is
sufficient to ensure $x(t) \in \mathcal{V}_d(t)$, $t \ge t_0$. This is
yet another significant motivation to investigate event-triggered
designs under bounded data rate constraints.

\section{Event-Triggered Control with Bounded Bit Rates and
  Instantaneous Transmission}\label{sec:ET_bnd_dat_rate}

In this section, we seek to design event-triggered laws for deciding
the transmission times and the number of bits used per transmission
based on feedback. We achieve this by letting the encoder at the
sensor, which has access to the exact plant state, make these
decisions in an opportunistic fashion. Here, we consider the
simplified scenario of instantaneous communication and tackle the more
general case of non-instantaneous communication in the next section.

\subsection{Requirements on the Encoding
  Scheme} \label{sec:coding-scheme}

Here, we specify the basic requirements of the encoding scheme
essential for our purposes. Consider the system defined
by~\eqref{eq:dynamics-x-xe} where the controller state evolves
according to~\eqref{eqn:x_hat}. Assume that, at the beginning $t_0\in
\realnonnegative$ of the zoom in stage, the encoder and decoder have a
common knowledge of a constant $d_e(t_0)$ such that
$\Infnorm{x_e(t_0)} \leq d_e(t_0)$. Given this common knowledge, the
encoder and the decoder \textit{inductively} construct a signal
$d_e(.)$ such that $\Infnorm{x_e(t)} \leq d_e(t)$ is satisfied for all
$t \geq t_0$ as follows. First, note that as a consequence of
\eqref{eqn:enc_err_dyn}, we have that
\begin{equation*}
  x_e(t) = e^{A(t-t_k)} x_e(t_k) + \int_{t_k}^{t} e^{A(t-s)} v(s) 
  \mathrm{d}s ,
\end{equation*}
which in turn implies
\begin{align*}
  \Infnorm{x_e(t)} &\leq \Infnorm{e^{A(t-t_k)} x_e(t_k)} +
  \int_{t_k}^{t} \Enorm{e^{A(t-s)} v(s) } \mathrm{d} s
  \\
  &\leq \Infnorm{e^{A(t-t_k)}} \Infnorm{x_e(t_k)} + \int_{t_k}^{t}
  e^{\Enorm{A}(t-s)} \nu \mathrm{d} s ,
\end{align*}
where $\nu$ is the uniform bound on the disturbance $v$, 
cf.~\eqref{eqn:disturb_bound}. Now, assuming that the encoder 
and the decoder know $d_e(t_k) \geq 0$ at time $t_k$ such that
$\Infnorm{x_e(t_k)} \leq d_e(t_k)$, then both can compute 
\begin{subequations}\label{eqn:de_def}
  \begin{equation}\label{eqn:de_evolve}
    d_e(t) \triangleq \Infnorm{e^{A(t-t_k)}} d_e(t_k) + 
    \frac{\nu}{\Enorm{A}} [ e^{\Enorm{A} (t-t_k)} - 1 ],
  \end{equation}
  for $t \in [t_k, t_{k+1})$. The above discussion guarantees that
  $\Infnorm{x_e(t)} \leq d_e(t)$ for $t \in [t_k, t_{k+1})$. Next, at
  time $t_{k+1}$, if $n p_{k+1}$ is the number of bits used to
  quantize and transmit information, then the encoder and the decoder
  update the value of $d_e(t_{k+1})$ by the jump,
  \begin{align}\label{eqn:de_update}
    d_e(t_{k+1}) = \frac{1}{2^{p_{k+1}}} d_e(t_{k+1}^-).
  \end{align}
\end{subequations}
Assuming the quantization at time $t_k$ is such that
$\Infnorm{x_e(t_k)} \leq d_e(t_k)$ given $\Infnorm{x_e(t_k)} \leq
d_e(t_k^-)$, then it is straightforward to verify by induction that
the so constructed signal $d_e$ ensures $\Infnorm{x_e(t)} \leq d_e(t)$
for all $t \geq t_0$.

As an example, we next specify (up to the number of bits) an encoding
scheme that satisfies the above requirements.  Given $d_e(t_{k})$ such
that $\Infnorm{x_e(t_{k})} \leq d_e(t_{k})$, for $k \in
\integersnonnegative$, the plant state satisfies
\begin{equation*}
  x(t) \in S(\hat{x}(t), d_e(t)) = \{ \xi \in \mathbb{R}^n :
  \Infnorm{\xi - \hat{x}(t)} \leq d_e(t) \} ,
\end{equation*}
for all $t \in [t_k, t_{k+1})$. At time $t_{k+1}$, the sensor/encoder
encodes the plant state and transmits using $n p_{k+1}$ bits. In this
encoding scheme, the set $S(\hat{x}(t_{k+1}^-), d_e(t_{k+1}^-))$ is
divided uniformly into $2^{n p_{k+1}}$ hypercubes and
$\hat{x}(t_{k+1})$ is chosen as the centroid of the hypercube
containing the plant state $x(t_{k+1})$. This results in
$d_e(t_{k+1})$ being updated as in~\eqref{eqn:de_update}. Formally, we
can express the quantization at time $t_k$ as
\begin{equation}\label{eqn:quant_def}
  q_k(x(t_k), \hat{x}(t_k^-)) \in \argmin_{\xi \in \Xc_k} \{ \Infnorm{x(t_k) - 
    \xi} \} ,
\end{equation}
where $\Xc_k$ is the set of centroids of the $2^{n p_k}$ hypercubes
that the set $S(\hat{x}(t_{k}^-), d_e(t_k^-))$ is divided into. We
assume that if $x(t_k)$ lies on the boundary of two or more
hypercubes, then the encoder and decoder choose the value of
$q_k(x(t_k), \hat{x}(t_k^-))$ according to a common deterministic
rule. As a result, given $\hat{x}(t_0)$ and $d_e(t_0)$ at time $t_0$,
$\hat{x}(t)$ and $d_e(t)$ are known to both the encoder and the
decoder at all times $t \geq t_0$.

In the remainder of the paper, we make no reference to this specific
encoding scheme. Instead it is sufficient for us to use the properties
of the encoding scheme specified by~\eqref{eqn:de_def}.

\subsection{Analysis of the Performance
  Ratio}\label{sec:perform-ratio}

We define the \emph{performance ratio} function, measuring
the ratio of the quadratic Lyapunov function $V$ and the desired
performance~$V_d$,
\begin{equation}\label{eqn:bt_def}
  b(t) \triangleq \frac{V(x(t))}{V_d(t)} .
\end{equation}
We use this function to determine the transmission times in an 
opportunistic fashion. First, however, we find it useful to
encapsulate some general properties of the performance ratio, $b(t)$,
and of its evolution as we use these properties through out the paper.

In the sequel, we make the following assumptions.
\begin{subequations}
  \begin{align}
    W & \triangleq \frac{\lambda_m(Q)}{\lambda_M(P)} - a \beta > 0 ,
    \label{eq:W}
    \\
    \sqrt{V_0} & \geq \frac{2 \Enorm{P} \nu}{ \sigma (a-1) \beta
      \sqrt{\lambda_m(P)} } , \label{eqn:V0_condition}
  \end{align}
\end{subequations}
where $a > 1$ and $\sigma \in (0, 1)$ are arbitrary constants.
Assumption~\eqref{eq:W} is sufficient to guarantee with
continuous-time and unquantized feedback a convergence rate faster
than $\beta$, in the absence of external disturbance.
Assumption~\eqref{eqn:V0_condition} prescribes an upper bound on the
norm of the tolerable disturbance given~$V_0$ (the steady state value
of $V_d$), or conversely prescribes $V_0$ given $\nu$.  This
interpretation becomes clearer later in the proofs of our results.

The following result provides an upper bound on the value of $b$ that
is convenient for our purposes. The proof can be found in the
appendix.

\begin{lemma}\longthmtitle{Upper Bound on Performance
    Ratio}\label{lem:bound-b}
  Given $t_k \in \realpositive$ such that $b(t_k) \leq 1$, then
  \begin{align*}
    b(\tau + t_k) \leq \tilde{b}(\tau, b(t_k), \epsilon(t_k)),
  \end{align*}
  for $\tau \ge 0$, where
  \begin{align}
    \epsilon(t) & \triangleq \frac{d_e(t)}{c \sqrt{V_d(t)}}, \quad
    \tilde{b}(\tau, b_0, \epsilon_0) \triangleq \frac{ f_1(\tau, b_0,
      \epsilon_0) }{ f_2(\tau) } , \label{eqn:eps_btild_def}
    \\
    f_1(\tau, b_0, \epsilon_0) &\triangleq b_0 + \frac{W \epsilon_0 }{
      w + \theta } ( e^{(w + \theta) \tau} - 1 ) + \frac{c_1 - c_2}{w}
    ( e^{w \tau} - 1 ) \notag
    \\
    & \quad + \frac{c_2}{ w + \Enorm{A} } ( e^{(w + \Enorm{A}) \tau} -
    1 ) , \notag
    \\
    f_2(\tau) &\triangleq e^{w \tau}, \notag
  \end{align}
  with $ w \triangleq \tfrac{\lambda_m(Q)}{\lambda_M(P)} - \beta >0 $,
  $\theta \triangleq \Enorm{A} + \tfrac{\beta}{2}$ and
  \begin{equation*}
    c \triangleq \frac{W \sqrt{ \lambda_m(P)} }{ 2 \sqrt{n}
      \Enorm{PBK} }, \ \ 
    c_1 \triangleq \frac{2 \Enorm{P}}{\sqrt{\lambda_m(P)}} 
    \frac{\nu}{\sqrt{V_0}}, \ \ 
    c_2 \triangleq \frac{W}{c \Enorm{A}} \frac{\nu}{\sqrt{V_0}}.
  \end{equation*} \qed
\end{lemma}

Motivated by Lemma~\ref{lem:bound-b}, we formally define the function
\begin{align}\label{eqn:Gamma1_def}
  \Gamma_1(b_0, \epsilon_0) \triangleq \min \{ \tau \geq 0 :
  \tilde{b}(\tau, b_0, \epsilon_0) = 1, \ \frac{\mathrm{d}
    \tilde{b}}{\mathrm{d} \tau} \geq 0 \} .
\end{align}
Thus, $\Gamma_1(b_0, \epsilon_0)$ is a lower bound on the time it
takes $b$ to evolve to $1$ starting from $b(t_k) = b_0$ with
$\epsilon(t_k) = \epsilon_0$.  The following result captures some
useful properties of this function, the proof of which can be found in
the appendix.

\begin{lemma}\longthmtitle{Properties of the Function
    $\Gamma_1$}\label{lem:Gamma1_prop}
  The following holds true,
  \begin{enumerate}
  \item $\Gamma_1(1, 1) > 0$.
  \item If $b_1 \geq b_0$ and $\epsilon_1 \geq \epsilon_0$, then
    $\Gamma_1(b_0, \epsilon_0) \geq \Gamma_1(b_1, \epsilon_1)$. In
    particular, if $b_0 \in [0, 1]$, then $\Gamma_1(b_0, \epsilon_0)
    \geq \Gamma_1(1, \epsilon_0)$.
  \item For $T>0$, if $b_0 \in [0, 1]$ and
    \begin{equation}\label{eqn:rho_def}
      \epsilon_0 \leq \rhofun{T}{b_0} \triangleq \frac{ (w + \theta) 
        ( 1 - b_0 )  }{ W ( e^{ (w+\theta) T } - 1 ) } + 1 ,
    \end{equation}
    then $\Gamma_1(b_0, \epsilon_0) \geq \min \{ \Gamma_1(1, 1), T
    \}$.
  \end{enumerate} \qed
\end{lemma}

\subsection{Event-Triggered Design with Arbitrary Finite 
  Communication Rate} \label{sec:instantaneous-bound}

Here, we solve the problem stated in Section~\ref{sec:prob_stat} in a
way that guarantees that the number of bits at each transmission is
finite, although not necessarily uniformly upper bounded across all
transmissions.  We build on these developments in
Section~\ref{sec:instantaneous-bound-uniform} to address the problem
when there exists an explicit uniform bound across all transmissions.

\begin{theorem}\longthmtitle{Control with Arbitrary Finite
    Communication Rate} \label{thm:inst_tx_fin_bit}
  Consider the system~\eqref{eqn:plant_dyn} under the feedback law $u
  = K \hat{x}$, with $t\mapsto \hat{x}(t)$ evolving according
  to~\eqref{eqn:x_hat} and the sequence $\{t_k\}_{k \in
    \integersnonnegative}$ determined recursively by
  \begin{equation}\label{eqn:tk_recursive}
    t_{k+1} = \min \setdef{ t \geq t_k}{ b(t) \geq 1, \
      \dot{b}(t) \geq 0} .
  \end{equation}
  Assume the encoding scheme is such that~\eqref{eqn:de_def} holds for
  all $t \ge t_0$. Further assume that $V(x(t_0)) \leq V_d(t_0)$ and
  that~\eqref{eq:W}-\eqref{eqn:V0_condition} hold. If the number of
  bits $np_k$ transmitted at time $t_k$ satisfies
  \begin{equation}\label{eqn:pk_lbound}
    p_k \geq \uline{p_k} \triangleq \Bigg\lceil \log_2 \left(
      \frac{d_e(t_k^-)}{c \sqrt{V_d(t_k)}} \right) \Bigg\rceil ,
  \end{equation}
  where recall $ c = \frac{W \sqrt{\lambda_m(P)}}{2 \sqrt{n}
    \Enorm{PBK}}$. Then the following holds:
  \begin{enumerate}
  \item the inter-transmission times $\{ T_k \}_{k \in
      \integerspositive} \triangleq \{ t_{k+1} - t_k \}_{k \in
      \integerspositive}$ have a uniform positive lower bound,
  \item the origin is exponentially practically stable for the
    closed-loop system, with $V(x(t)) \leq V_d(t)$ for all $t \geq
    t_0$.
  \end{enumerate}
\end{theorem}
\begin{proof}
  From~\eqref{eqn:tk_recursive} note that $b(t_k) = 1$. Also,
  from~\eqref{eqn:pk_lbound}, \eqref{eqn:de_update} and the definition
  of $\epsilon(t)$ in~\eqref{eqn:eps_btild_def} we see that
  $\epsilon(t_k) \leq 1$. Thus, as a consequence of
  Lemma~\ref{lem:Gamma1_prop}, we see that for any $k \in
  \integerspositive$, $T_k \geq \Gamma_1(b(t_k), \epsilon(t_k)) \geq
  \Gamma_1(1,1) > 0$. The second claim of the theorem follows from the
  fact that~\eqref{eqn:tk_recursive} ensures $b(t) \leq 1$ for all $t
  \geq t_0$.
\end{proof}

The idea behind the trigger~\eqref{eqn:tk_recursive} is to let the
system evolve until the performance criterion is about to be violated
(performance ratio $b$ about to exceed $1$) and only then transmit to
close the feedback loop. The quantity $n \uline{p_k}$ in
Theorem~\ref{thm:inst_tx_fin_bit} can be interpreted as the
``minimum''\footnote{We use ``minimum'' here in the context of the
  encoding scheme of Section~\ref{sec:coding-scheme} and other design
  choices and approximations made in the paper.} number of bits that
would ensure $\dot{b} < 0$ just after transmissions. Incidentally,
this condition also ensures $\epsilon(t_k) \leq 1$, which in turn
guarantees that, after transmission, $b < 1$ for at least the next
$\Gamma_1(1,1)$ units of time (Lemma~\ref{lem:Gamma1_prop}). The
recursive nature of the inequalities~\eqref{eqn:pk_lbound} can be
leveraged to better understand the relationship across different times
among the bounds on the number of bits sufficient for stability. In
order to provide an intuitive interpretation, we assume in the
following result that there is no disturbance in the system ($\nu = 0$
and $V_0 = 0$). The result gives insight into the total number of
bits sufficient for stability as a function of time.

\begin{corollary}\longthmtitle{Upper Bound on the Data  Rate Sufficient
    for Stability}\label{cor:sufficient-bit-rate}
  Under the assumptions of Theorem~\ref{thm:inst_tx_fin_bit} and no
  disturbances, the following holds for any $k \in \integerspositive$,
  \begin{align*}
    % \uline{\Bc}(t_k,t_0) \triangleq
    & n ( \uline{p_k} + \sum_{i=1}^{k-1} p_i )
    \\
    &\leq n \Big( \Infnorm{A} + \frac{\beta}{2} \Big) \log_2(e) (t_k -
    t_0) \!+\! n \log_2 \bigg( \frac{ d_e(t_0)} { c \sqrt{V_d(t_0)} }
    \bigg) \!+\! n .
  \end{align*}
\end{corollary}
\begin{proof}
  Using~\eqref{eqn:de_def} (with $\nu = 0$ and $V_0 = 0$) recursively
  gives
  \begin{align*}
    d_e(t_k^-) &= \Infnorm{e^{A T_{k-1}}} d_e(t_{k-1}) =
    \frac{\Infnorm{e^{A T_{k-1}}} d_e(t_{k-1}^-)}{2^{p_{k-1}}}
    \\
    &= \prod_{i=1}^{k-1} \frac{\Infnorm{e^{A T_i}}}{2^{p_i}}
    \Infnorm{e^{A T_0}} d_e(t_0)
    \\
    &\leq \frac{ e^{\Infnorm{A} (t_k - t_0) } }{ \prod_{i=1}^{k-1} 2^{
        p_i } } d_e(t_0) ,
  \end{align*}
  for $k \in \integerspositive$. Substituting this bound
  in~\eqref{eqn:pk_lbound} (and multiplying by $n$ to give us the
  number of bits), we arrive at
  \begin{align*}
    & n \uline{p_k} \leq n \Bigg\lceil \log_2 \left( \frac{
        e^{\Infnorm{A}(t_k - t_0)} d_e(t_0)} { e^{\frac{- \beta}{2}(t_k
          - t_0)} c \sqrt{V_d(t_0)} } \right) \Bigg\rceil - n
    \sum_{i=1}^{k-1} p_i,
  \end{align*}
  where we have used $V_d(t_k) = V_d(t_0) e^{- \beta (t_k - t_0)}$.
  Upper bounding $\lceil . \rceil$ and rearranging the terms yields
  the result.
\end{proof}

\begin{remark}\longthmtitle{Observations about
    Corollary~\ref{cor:sufficient-bit-rate}}
  {\rm Corollary~\ref{cor:sufficient-bit-rate} is interesting for the
    following reasons:
  \begin{itemize}
  \item The upper bound on the sufficient number of bits to be
    transmitted up to time $t_k$, for any $k \in \integerspositive$,
    depends only on the length of the time interval $t_k - t_0$, the
    initial conditions $d_e(t_0)$ and $V_d(t_0)$ and the system
    parameters. Thus the sufficient data rate is uniformly bounded;
  \item If more bits than sufficient are transmitted in the past,
    ($p_i > \uline{p_i}$ for some $i < k$), then fewer bits are
    sufficient at $t_k$;
  \item The expression, albeit only being valid at the transmission
    times $\{t_k\}_{k \in \integersnonnegative}$, has a form similar
    to the lower bound~\eqref{eqn:Bc_t} on the number of bits
    transmitted over the time interval $[t_0,t]$ in
    Proposition~\ref{prop:necs_rate}. In fact, the occurrence of
    $\Infnorm{A}$ in Corollary~\ref{cor:sufficient-bit-rate} is a
    by-product of the use of the norm $\Infnorm{.}$ and hypercubes as
    our quantization domains. In comparison with~\eqref{eqn:Bc_t}, $n
    \Infnorm{A}$ plays the role of $\trace{A}$, and $d_e(t_0)^n$ is
    proportional to~$\volume{E(t_0)}$ and we see that in the scalar
    case ($n= 1$) the sufficient asymptotic data rate is the same as
    the necessary asymptotic data rate;
  \item Theorem~\ref{thm:inst_tx_fin_bit} does not provide a uniform
    bound on $\uline{p_k}$. However (at least in the absence of
    disturbance), since the data rate is uniformly bounded, one can
    deduce that for any $k \in \integerspositive$, if $t_k - t_{k-1}$
    is bounded, then so is $\uline{p_k}$. \oprocend
  \end{itemize} 
}
\end{remark}

\subsection{Event-Triggered Design with Uniform Bound on 
Communication Rate} \label{sec:instantaneous-bound-uniform}

In this section, we expand on our previous discussion to solve the
problem stated in Section~\ref{sec:prob_stat} with a uniform bound on
the number of bits per transmission. This is particularly relevant in
cases where the communication channel imposes a hard bound, say $n
\bar{p}$, on the number of bits that can be transmitted at each time.
Before getting into the technical details, we briefly lay out the
rationale behind our design. As a consequence of the hard limit on the
channel capacity, a transmission at a time~$t_k \in \realpositive$ can
be caused by either of the following two reasons:
\begin{enumerate}
\item[(Ti)] the system trajectory hits the limit of the required
  performance guarantee, i.e., $b(t_k)=1$, as
  in~\eqref{eqn:tk_recursive}, or
\item[(Tii)] even though $b(t_k)<1$, the number of bits required later
  to keep $b$ from exceeding $1$ would be larger than the ``channel
  capacity'' $n \bar{p}$.
\end{enumerate}
To design an appropriate trigger for (Tii), we make use of
Lemma~\ref{lem:Gamma1_prop}, which characterizes the time it takes $b$
to evolve from any value to $1$. This information allows us to
determine the ``minimum" number of bits to be transmitted so that $b$ takes
at least a certain pre-designed time to reach $1$.  Our trigger for
(Tii) would then be simply `transmit if this ``minimum" number of bits
reaches the maximum channel capacity'.

\subsubsection*{Trigger Design and Analysis}

The analysis of Section~\ref{sec:perform-ratio} sets the basis for
computing the ``minimum" number of bits that guarantee that the
performance specification is met for a certain pre-designed
time. Specifically, define the \emph{channel-trigger} function
\begin{equation}\label{eqn:h2_def}
  \triggerCh(t) \triangleq  \frac{ \epsilon(t) }{ \rhofun{T}{b(t)} }
  = \frac{d_e(t)}{c \sqrt{V_d(t)} \rhofun{T}{b(t)}}  ,
\end{equation}
where $T > 0$ is a fixed design parameter.
Lemma~\ref{lem:Gamma1_prop}(iii) implies that, if $\triggerCh(t_k)
\leq 1$, then $b(t) \leq 1$ for at least $t \in [t_k,
  t_k + \min \{ T, \Gamma_1(1,1) \} )$.
% Thus, our idea is to transmit whenever $b = 1$ as
% in~\eqref{eqn:tk_trig2} or
Building on this observation, our trigger for (Tii) is then transmit
if $\displaystyle {\triggerCh(t)}/{2^{\bar{p}}} = 1$, i.e., when `the
number of bits required to have the value of $\triggerCh$ smaller than
or equal to $ 1$ just after transmission' is no more than $n \bar{p}$,
the upper bound imposed by the channel.

The next result provides an upper bound on the function $\triggerCh$
and is useful later when establishing a uniform lower bound on the
inter-transmission times for our design.

\begin{lemma}\longthmtitle{Upper Bound on Channel-Trigger
    Function}\label{lem:bound-triggerCh} 
  Given $t_k \in \realpositive$ such that $b(t_k) \leq 1$, then
  \begin{align*}
    \triggerCh(\tau + t_k) \leq \triggerChbar( \tau, b(t_k),
    \epsilon(t_k), \epsilon(t_k ) ) ,
  \end{align*}
  for $\tau \ge 0$, where
  \begin{align}\label{eqn:hbar_def}
    & \triggerChbar(\tau, b_0, \epsilon_0, \psi_0) \notag
    \\
    &\triangleq \frac{\Infnorm{e^{A \tau}} e^{\frac{\beta}{2} \tau}
      \psi_0 }{ \rhofun{T}{\tilde{b}(\tau, b_0, \epsilon_0)} } +
    \frac{ \nu ( e^{\Enorm{A} \tau} - 1 ) }{ c \Enorm{A}
      \rhofun{T}{\tilde{b}(\tau, b_0, \epsilon_0)} \sqrt{ V_0 } } .
  \end{align}
\end{lemma}
\begin{proof}
  From its definition, we can bound $\triggerCh$
  using~\eqref{eqn:de_evolve}, the fact that $\rho_{T}$ is a
  decreasing function and Lemma~\ref{lem:bound-b} as,
  \begin{align*}
    \triggerCh(\tau + t_k) \leq \frac{\Infnorm{e^{A \tau}} d_e(t_k) +
      \frac{\nu}{\Enorm{A}} ( e^{\Enorm{A} \tau} - 1 )}{c
      \rhofun{T}{\tilde{b}(\tau, b(t_k), \epsilon(t_k)} \sqrt{ V_d(\tau+
        t_k) }}.
  \end{align*}
  The result now follows by further simplifying this expression
  expanding $V_d(\tau+ t_k) = V_d(t_k) e^{-\beta \tau} + V_0 ( 1 -
  e^{-\beta \tau} )$, observing that $V_0 \geq 0$ and $V_d(t) \geq
  V_0$ for all $t \geq t_0$, and using the definition of $\epsilon$.
\end{proof}

Given Lemma~\ref{lem:bound-triggerCh}, we define the function
\begin{align*}
  \Gamma_2(b_0, \epsilon_0, \psi_0) &\triangleq \min \{ \tau \geq 0 :
  \frac{ \triggerChbar(\tau, b_0, \epsilon_0, \psi_0) }{ 2^{\bar{p}} }
  = 1 \},
\end{align*}
which is a lower bound on the time it takes $\triggerCh(\tau + t_k)$
to reach $2^{\bar{p}}$ given $b(t_k) = b_0$ and $\epsilon(t_k) =
\epsilon_0$.  Note that the argument $\psi_0$ in the definitions of
$\triggerChbar$ and $\Gamma_2$ is redundant for our purposes here, but
will play an important role later when discussing the case of
non-instantaneous communication.
% Nevertheless, we introduce $\psi_0$ at this stage itself in order to
% avoid redefinitions.

We are now ready to present the main result of this section.

\begin{theorem}\longthmtitle{Control under Bounded Channel
    Capacity}\label{thm:inst_tx_bounded_bit}
  Consider the system~\eqref{eqn:plant_dyn} under the feedback law $u
  = K \hat{x}$, with $t\mapsto \hat{x}(t)$ evolving according
  to~\eqref{eqn:x_hat} and the sequence $\{t_k\}_{k \in
    \integersnonnegative}$ determined recursively by
  \begin{multline}\label{eqn:tk_trig2}
    t_{k+1} = \min \{ t \geq t_k : b \geq 1, \ \dot{b}(t) \geq 0 \;
    \operatorname{OR} \; \frac{ \triggerCh(t) }{ 2^{\bar{p}} } \geq 1
    \} ,
  \end{multline}
  where $n \bar{p}$ is the upper bound on the number of bits that can
  be sent per transmission and $T > 0$ in the
  definition~\eqref{eqn:h2_def} of $\triggerCh$ is a design
  parameter. Assume the encoding scheme is such
  that~\eqref{eqn:de_def} is satisfied for all $t \ge t_0$. Further
  assume that $V(x(t_0)) \leq V_d(t_0)$, $\triggerCh(t_0) \leq
  2^{\bar{p}}$ and that~\eqref{eq:W}-\eqref{eqn:V0_condition}
  hold. Let $\uline{p_k}$ be given by
  \begin{equation}\label{eqn:pk_lbound_case2}
    \uline{p_k} \triangleq \Bigg\lceil \log_2 \left(
      \frac{d_e(t_k^-)}{c \rhofun{T}{b(t_k)} \sqrt{V_d(t_k)}} \right) 
    \Bigg\rceil,
  \end{equation}
  where recall $ c = \frac{W \sqrt{\lambda_m(P)}}{2 \sqrt{n}
    \Enorm{PBK}}$. Then, the following hold:
  \begin{enumerate}
  \item $\uline{p_1} \leq \bar{p}$. Further for each $k \in
    \integerspositive$, if $ p_k \in \integerspositive \intersect [
    \uline{p_k}, \bar{p} ]$, then $\uline{p_{k+1}} \leq \bar{p}$.
  \item the inter-transmission times $\{ T_k = t_{k+1} - t_k \}_{k \in
      \integerspositive}$ have a uniform positive lower bound,
  \item the origin is exponentially practically stable for the
    closed-loop system, with $V(x(t)) \leq V_d(t) = ( V_d(t_0) - V_0)
    e^{ -\beta (t - t_0)} + V_0$ for all $t \geq t_0$.
  \end{enumerate}
\end{theorem}
\begin{proof}
  Since $V(x(t_0)) \leq V_d(t_0)$ and $\triggerCh(t_0) \leq
  2^{\bar{p}}$, the trigger~\eqref{eqn:tk_trig2} implies that
  $\uline{p_1} \leq \bar{p}$. Similarly, if for each $k \in
  \integerspositive$, $p_k \in \integerspositive \intersect [
  \uline{p_k}, \bar{p} ]$, then~\eqref{eqn:tk_trig2} implies
  $\uline{p_{k+1}} \leq \bar{p}$, which proves~\emph{(i)}.

  To show~\emph{(ii)}, we study each of the two conditions that
  define~\eqref{eqn:tk_trig2}. Regarding the condition on the
  performance-ratio function, note that $\Gamma_1(b(t_k),
  \epsilon(t_k))$ is, by definition, a lower bound on the time it
  takes the condition to be enabled. Since~\eqref{eqn:tk_trig2}
  guarantees that $\triggerCh(t_k^-) \leq 2^{\bar{p}}$ and, as a
  result, $\triggerCh(t_k) \leq 1$ (with equality holding when $p_k =
  \uline{p_k}$), we have $\epsilon(t_k) \leq
  \rhofun{T}{b(t_k)}$. Therefore, Lemma~\ref{lem:Gamma1_prop}
  guarantees that $\Gamma_1(b(t_k), \epsilon(t_k)) \geq \min \{
  \Gamma_1(1, 1), T \} > 0$ for $k \in \integerspositive$.  Regarding
  the condition on the channel-trigger function
  in~\eqref{eqn:tk_trig2}, note that $\Gamma_2(b(t_k), \epsilon(t_k),
  \epsilon(t_k) )$ is, by definition, a lower bound on the time it
  takes the condition to be enabled.  We therefore focus on upper
  bounding the function $\triggerChbar$ that
  defines~$\Gamma_2$. First, notice that for $b_0\le 1$ and
  $\epsilon_0 \leq \rhofun{T}{b_0}$,~\eqref{eq:tildeb-ineq} implies
  that $\tilde{b}(\tau, b_0, \epsilon_0 ) \leq \tilde{b}(\tau, 1, 1 )$
  for all $\tau \in [ 0, \min \{ \Gamma_1(1, 1), T \} ]$.  The fact
  that $\rho_T$ is decreasing then implies that the second term in the
  definition~\eqref{eqn:hbar_def} of $\triggerChbar$ can be bounded
  by,
  \begin{align*}
    &\frac{ \nu ( e^{\Enorm{A} \tau} - 1 )/c }{\Enorm{A}
      \rhofun{T}{\tilde{b}( \tau, b_0, \epsilon_0)} \sqrt{ V_0 } }
    \leq \phi_2(\tau) \triangleq \frac{ \nu ( e^{\Enorm{A} \tau} - 1 )/c
    }{\Enorm{A} \rhofun{T}{ \tilde{b}( \tau, 1, 1) } \sqrt{ V_0 } }
    ,
  \end{align*}
  for $\tau \in [ 0, \min \{ \Gamma_1(1, 1), T \} ]$.  Next, we turn
  our attention to the first term in the
  definition~\eqref{eqn:hbar_def} of $\triggerChbar$.  Let $c_3$ be
  the negative of the coefficient of $b_0$ in the
  definition~\eqref{eqn:rho_def} of $\rhofun{T}{b_0}$. Observe that
  for $b_0 \geq 0$, $\epsilon_0 \geq 0$ and $\tau \in [ 0, \min \{
  \Gamma_1(1, 1), T \} ]$,
  \begin{align*}
    &\frac{\mathrm{d}}{\mathrm{d} \tau} \frac{ \psi_0 }{ \rhofun{T}{
        \tilde{b}(\tau, b_0, \epsilon_0) } }
    \\
    &= \frac{ \psi_0 c_3 }{ \rhofun{T}{ \tilde{b}(\tau, b_0,
        \epsilon_0) }^2 } [ - w \tilde{b} + W \epsilon_0 e^{\theta
      \tau} + c_1 + c_2 ( e^{\Enorm{A} \tau} - 1 ) ]
    \\
    &\leq \psi_0 c_3 [ W \epsilon_0 e^{\theta \tau} + c_1 + c_2 (
    e^{\Enorm{A} \tau} - 1 ) ] ,
  \end{align*}
  where we have used~\eqref{eqn:btilde_diff_eq} and the facts that
  $\tilde{b}(\tau, b_0, \epsilon_0 ) \le \tilde{b}(\tau, 1, 1) \le 1$
  for all $\tau \in [ 0, \min \{ \Gamma_1(1, 1), T \} ]$ and
  $\rhofun{T}{y} \geq 1$ for $y \in [0, 1]$. Then, the Comparison
  Lemma~\cite{HKK:02} implies that
  \begin{align*}
    &\frac{ \psi_0 }{ \rhofun{T}{ \tilde{b}(\tau, b_0, \epsilon_0) } }
    \leq \frac{ \psi_0 }{ \rhofun{T}{b_0} } \ +
    \\
    &\quad \psi_0 c_3 \bigg[ \frac{W \epsilon_0}{\theta} ( e^{\theta
      \tau} - 1) + \frac{c_2}{\Enorm{A}} ( e^{\Enorm{A} \tau} - 1 ) +
    (c_1 - c_2) \tau \bigg].
  \end{align*}
  Defining now $ \phi(\tau, \phi_0) \triangleq \Infnorm{ e^{A \tau} }
  e^{ \frac{\beta}{2}\tau } \phi_1(\tau, \phi_0) + \phi_2(\tau) $,
  with
  \begin{align*}
    &\phi_1( \tau, \phi_0 ) \triangleq \phi_0 + \phi_0 \rhofun{T}{0} 
    c_3 \bigg[ \frac{W \rhofun{T}{0} }{\theta} ( e^{\theta \tau} - 1) 
    \ +
    \\
    &\qquad\qquad\qquad\qquad \frac{c_2}{\Enorm{A}} ( e^{\Enorm{A} 
      \tau} - 1 ) + (c_1 - c_2) \tau \bigg].
  \end{align*}
  we deduce, for $\epsilon_0 \leq \rhofun{T}{b_0}$ and $\tau \in [ 0,
  \min \{ \Gamma_1(1, 1), T \} ]$,
  \begin{align}\label{eqn:upper-bound-triggerChbar}
    \triggerChbar(\tau, b_0, \epsilon_0, \psi_0) \leq \phi \big(\tau,
    \frac{ \psi_0 }{ \rhofun{T}{b_0} } \big) ,
  \end{align}
  where we have used $\rhofun{T}{ b_0 } \leq \rhofun{T}{0}$. Note 
  that since we are interested in lower bounding $\Gamma_2( 
  b(t_k), \epsilon(t_k), \epsilon(t_k) )$ with $\epsilon(t_k) \leq 
  \rhofun{T}{ b(t_k) }$, we can focus on the case $\psi_0 = 
  \epsilon_0 \leq \rhofun{T}{b_0}$, which leads to the bound
  \begin{align*}
  \triggerChbar(\tau, b_0, \epsilon_0, \psi_0) \leq \phi \big(\tau,
  1 \big) .
  \end{align*}
  Thus $\triggerChbar$ is bounded by a function that depends only on
  $\tau$ and is equal to $1$ at $\tau = 0$.  Hence, we deduce the
  existence of a uniform positive lower bound on the function
  $\Gamma_2(b_0, \epsilon_0, \psi_0)$ for $b_0 \in [0, 1]$ and $\psi_0
  = \epsilon_0 \leq \rhofun{T}{b_0}$. Thus $T_k = t_{k+1} - t_k \geq
  \min \{ T, \Gamma_1(b(t_k), \epsilon(t_k)) , \Gamma_2(b(t_k),
  \epsilon(t_k), \epsilon(t_k) \}$, for $k \in \integerspositive$ has
  a uniform positive lower bound, proving~\emph{(ii)}.  Claim
  \emph{(iii)} follows by noting that \emph{(i)} and \emph{(ii)} imply
  $b(t) \leq 1$,~$t \geq t_0$.
\end{proof}

The quantity $n \uline{p_k}$ in Theorem~\ref{thm:inst_tx_bounded_bit}
has now a slightly different interpretation than in
Theorem~\ref{thm:inst_tx_fin_bit}: it corresponds to the ``minimum"
number of bits sufficient to ensure that, after transmission, $b$
remains less than $1$ for at least the next $\min \{ T,
\Gamma_1(b(t_k), \epsilon(t_k))$ units of time.

\section{Event-Triggered Control with Bounded Bit Rates and
  Non-Instantaneous Transmission}\label{sec:ET_bnd_dat_rate-non-inst}

Here we design event-triggered laws for deciding the transmission
times and the number of bits used per transmission when communication
is not instantaneous. Such scenarios are common when the model
available for the communication channel specifies a capacity in terms
of bit rates.  In this case, we need to distinguish between the time
when the encoder/sensor transmits from the time when the
decoder/controller receives a complete packet of data. This
corresponds to the setup of Section~\ref{sec:prob_stat} in its full
generality.

\subsection{Information Consistency Between Encoder and
  Decoder}\label{sec:information-consistency}

Given the difference between transmission and communication times, the
first problem we tackle is making sure that the information (the state
estimate $\hat{x}$ and the upper bound $d_e$ on the encoding error
$x_e$) used by the encoder and the decoder is consistent.  The
mechanisms described here rely critically on the assumptions of
synchronized clocks and common knowledge of the communication time,
cf.  Section~\ref{sec:prob_stat}. According to the problem statement,
the encoder encodes its message at~$t_k$ and sends $n p_k$ bits which
are received completely by the decoder at~$r_k \geq t_k$.
% Now we describe the algorithms employed by the encoder and the
% decoder to update the controller and coding states.
Algorithms~\ref{algo:enc_update} and~\ref{algo:dec_update} describe,
respectively, how the encoder and the decoder update $\hat{x}$ and
$d_e$ synchronously at the time instants $r_k$.

\begin{algorithm}[htb]
  {
    \footnotesize \vspace*{1ex} 
    \textbf{At $t = t_0 = r_0$, the encoder initializes}
    \begin{algorithmic}[1]
    	\STATE  $\delta_0 \gets d_e(t_0)$
    	\mycomment{store initial bound on encoding error}
    \end{algorithmic}
    
    \vspace{5pt}
    \textbf{At $t \in \{ t_k \}_{k \in \integerspositive}$, the 
    encoder sets}
    \begin{algorithmic}[1]
    	\setcounter{ALC@line}{1}	  	
      \STATE $z_k \gets \hat{x}(t_k^-)$ 
      \mycomment{store encoder variable}
      \STATE  $z_{E, k} \gets q_{E,k}( x(t_k), z_k )$
      \\
      \mycomment{encode plant state with $p_k$ bits}
      \STATE  $\displaystyle \delta_k \gets {d_e(t_k^-)}/{2^{p_k}}$
      \mycomment{compute bound on encoding error}
    \end{algorithmic}
    
    \vspace{5pt}
    \textbf{At $t \in \{ r_k \}_{k \in \integerspositive}$, the 
    encoder sets}
    \begin{algorithmic}[1]
      \setcounter{ALC@line}{4}
      \STATE $z_{D,k} \gets q_{D,k}( z_{E,k} , z_k )$
      \mycomment{decode plant state at $t_k$}
      \STATE $\hat{x}(r_k) \gets \!e^{\bar{A} \Delta_k} z_k \!+\! 
      e^{A \Delta_k} ( z_{D,k} - z_k )$ 
      \\
      \mycomment{update controller state}
      \STATE $d_e(r_k) \gets \Infnorm{e^{A \Delta_k}} \delta_k + 
      \frac{\nu}{\Enorm{A}} [ e^{\Enorm{A} \Delta_k } - 1 ]$
      \\
      \mycomment{update bound on encoding error}
    \end{algorithmic}	   
  }
  \caption{\hspace*{-.5ex}: Update of encoder
    variables}\label{algo:enc_update}
\end{algorithm}
 
\begin{algorithm}[htb]
  {
    \footnotesize \vspace*{1ex}
    \textbf{At $t = t_0 = r_0$, the decoder initializes}
    \begin{algorithmic}[1]
    	\STATE  $\delta_0 \gets d_e(t_0)$
    	\mycomment{store initial bound on encoding error}
    \end{algorithmic}
    
    \vspace{5pt}
    \textbf{At $t \in \{ r_k \}_{k \in \integerspositive}$, the 
    decoder sets}
    \begin{algorithmic}[1]
    	\setcounter{ALC@line}{1}
      \STATE $z_k \gets e^{-\bar{A} \Delta_k} \hat{x}(r_k^-)$
      \mycomment{compute encoder state at $t_k$}
      \STATE $z_{E,k}$
      \mycomment{received from the encoder}
      \STATE $ \delta_k \gets \frac{1}{2^{p_k}} \big( \Infnorm{e^{A 
      ( t_k^- - t_{k-1} ) } } \delta_{k-1} + \frac{\nu}{\Enorm{A}} [ 
      	e^{\Enorm{A} ( t_k^- - t_{k-1} ) } - 1 ] \big)$ 
      \\
      \mycomment{compute bound on encoding error at $t_k$}
      \STATE $z_{D,k} \gets q_{D,k}( z_{E,k} , z_k )$
      \mycomment{decode plant state at $t_k$}
      \STATE $\hat{x}(r_k) \gets \!e^{\bar{A} \Delta_k} z_k \!+\! 
      e^{A \Delta_k} ( z_{D,k} - z_k )$ 
      \\
      \mycomment{update controller state}
      \STATE $d_e(r_k) \gets \Infnorm{e^{A \Delta_k}} \delta_k + 
      \frac{\nu}{\Enorm{A}} [ e^{\Enorm{A} \Delta_k } - 1 ]$
      \\
      \mycomment{update bound on encoding error}

    \end{algorithmic}
  }
  \caption{\hspace*{-.5ex}: Update of decoder
    variables}\label{algo:dec_update}
\end{algorithm}

It is interesting to note that, as described above, the algorithms are
also applicable in the case of instantaneous communication. The idea
of Step 6 in each algorithm is to propagate $z_{D,k}$ forward in time
so that it may be used from time $r_k$ onwards (in the case of
instantaneous communication, note that $\hat{x}(r_k) = z_{D,k}$). We
next establish that Algorithms~\ref{algo:enc_update}
and~\ref{algo:dec_update} provide consistent signals $t \mapsto
\hat{x}(t)$, $t \mapsto d_e(t)$ to the encoder and the decoder, with
its proof in the appendix.

\begin{lemma}\longthmtitle{Consistency of
    Algorithms~\ref{algo:enc_update}
    and~\ref{algo:dec_update}}\label{lem:consistent_algos}
  If initially the encoder and the decoder share identical values for
  $\hat{x}(t_0)$ and $d_e(t_0)$ then Algorithms~\ref{algo:enc_update}
  and~\ref{algo:dec_update} result in consistent $\hat{x}(t)$ and
  $d_e(t)$ signals for all $t \geq t_0$. Further, $t\mapsto
  \hat{x}(t)$ evolves according to~\eqref{eqn:x_hat} and
  $\Infnorm{x_e(t)} \leq d_e(t)$ with $d_e(t)$ defined for $t \in
  [r_k, r_{k+1})$ for $k \in \integersnonnegative$ as
  \begin{subequations}\label{eqn:coding_scheme}
    \begin{align}
      d_e(t) &\triangleq \Infnorm{e^{A(t-t_k)}} \delta_k + 
      \frac{\nu}{\Enorm{A}} [ e^{\Enorm{A} (t-t_k)} - 1 ] ,
      \label{eqn:de_evolve_rk}
      \\
      \delta_{k+1} &= \frac{1}{2^{p_{k+1}}} d_e(t_{k+1}^-). 
      \label{eqn:de_jump_rk}
    \end{align}
  \end{subequations}	\qed
\end{lemma}

Note that although $d_e$ is updated by a jump at $\{ r_k \}_{k \in
  \integerspositive}$, the reference time in~\eqref{eqn:de_evolve_rk}
is still $t_k$ (because using the reference time $r_k$ instead would
result in a larger encoding error bound).

\subsection{Trigger Design and Analysis}

The basic underlying idea behind our event-triggered design in the
scenario of non-instantaneous communication is to anticipate ahead of
time the crossings of~$1$ by the performance-ratio function $b$ and the
channel-trigger function $\triggerCh$ after transmitting at most $n
\bar{p}$ number of bits.  Noting the update rule that gives $d_e(r_k)$
in Algorithms~\ref{algo:enc_update} and~\ref{algo:dec_update} and
following arguments analogous to those of
Lemma~\ref{lem:bound-triggerCh}, we see that
\begin{align*}
  & \triggerCh(r_k) \leq \triggerChbar \bigg( \Delta_k, b(t_k^-),
  \epsilon(t_k^-), \frac{ \epsilon(t_k^-) }{ 2^{p_k} } \bigg).
\end{align*}
Unlike in the case of instantaneous communication, we need to
distinguish between the third and the fourth argument
in~$\triggerChbar$ because the transmitted bits do not affect the
value of $\epsilon$ until $r_k$. If we can ensure that
$\triggerCh(r_k) \leq 1$, then the definition~\eqref{eqn:Gamma1_def}
of $\Gamma_1$ and Lemma~\ref{lem:Gamma1_prop} guarantee $b
\leq 1$ until $r_k + \min \{ \Gamma_1(1, 1), T \}$. To anticipate
$\triggerCh(r_k) \leq 1$, we define
\begin{align}\label{eqn:Gamma2_tilde_def}
  \tilde{\Gamma}_2(b_0, \epsilon_0, \psi_0) &\triangleq \min \{ \tau
  \geq 0 : \triggerChbar(\tau, b_0, \epsilon_0, \psi_0) = 1 \}.
\end{align}
From~\eqref{eqn:upper-bound-triggerChbar} we have that for $(
2^{\bar{p}} \psi_0 ) = \epsilon_0 \leq \rhofun{T}{b_0}$,
$\tilde{\Gamma}_2(b_0, \epsilon_0, \psi_0) \geq \min \{ \Gamma_1(1,
1), T, T^* \}$ with
\begin{equation*}
  T^* \triangleq \min \{ \tau \geq 0 : \phi \big( \tau, 1 / ( 
  2^{\bar{p}} ) \big) = 1 \}.
\end{equation*}
Given this discussion, we make the following assumption on the
function~$\Delta$ that describes the communication channel.
\begin{enumerate}[label={\textbf{(A)}},ref={A}]
\item For any $t \in \realnonnegative$, $\Delta(t,1) \geq 0$. Also, if
  $s_1 \leq s_2$, then $\Delta(t,s_1) \leq \Delta(t,s_2)$.  Given
  $\bar{p} \in \integerspositive$, there exists $T_M \in
  \realnonnegative$ with $ T_M < \min \{ \Gamma_1(1, 1), T, T^* \}$
  such that $\Delta(t,\bar{p}) \leq T_M$ for all $t \geq
  0$. \label{A:transmit_time} %\setcounter{saveenum}{\value{enumi}}
\end{enumerate}
Hence the event-triggering rule must anticipate at least $T_M$ units
of time ahead the crossing of $1$ by $b$ and anticipate
$\triggerCh(r_k) \geq 1$ even after having transmitted the maximum
number of bits, $n \bar{p}$, at $t_k$. In other words, we want to
ensure $\triggerCh(r_k) \leq 1$ so that $b < 1$ for at least all $t
\in [r_k, t_{k+1})$.  The fact $T_M < \min \{ \Gamma_1(1, 1), T, T^*
\}$ then ensures $t_{k+1} - r_k > 0$.

Our problem then reduces to checking the zero-crossing of the
functions $\Gamma_1 - T_M$, and $\tilde{\Gamma}_2 - T_M$. However,
computing the functions $\Gamma_1$ and $\tilde{\Gamma}_2$ repeatedly
as part of the event-triggering rule would impose an unnecessary
computational burden. For this reason, we seek a way to check the
conditions without having to explicitly compute $\Gamma_1$ and
$\tilde{\Gamma}_2$. The following result provides a solution for the
case of~$\Gamma_1$. We provide its proof in the appendix.

\begin{lemma} \longthmtitle{Algebraic Condition to Check if $b < 1$
    for the next $T^\circ$ units of time}\label{lem:Gamma1-T_sign} Let
  $T^\circ > 0$. For any $b_0 \in [0, 1]$, $\Gamma_1(b_0, \epsilon_0)
  > T^\circ$ if and only if $\tilde{b}(T^\circ, b_0, \epsilon_0) <
  1$. Further, the corresponding statement with the inequalities
  reversed and the one in which the inequalities are replaced by
  equality are true. \qed
\end{lemma}

Next, we make a similar observation about $\Gamma_2$. Again, we
provide the proof in the appendix.

\begin{lemma}\longthmtitle{Algebraic Condition to Check the Sign of
    $\tilde{\Gamma}_2 - T^\circ$}\label{lem:Gamma2-T_sign}
  Let $T^\circ > 0$. For any $b_0 \geq 0$ and $\epsilon_0 \geq 0$,
  $\tilde{\Gamma}_2(b_0, \epsilon_0, \psi_0) > T^\circ$ if and only if
  $\triggerChbar(T^\circ, b_0, \epsilon_0, \psi_0) < 1$.  Further, the
  corresponding statement with the inequalities reversed and the one
  in which the inequalities are replaced by equality are true. \qed
\end{lemma}

We are finally ready to present the main result of the section.

\begin{theorem}\longthmtitle{Bounded  Communication
    Rate with Non-Instantaneous
    Transmission}\label{thm:noninst_com_bounded_bit} 
  Consider the system~\eqref{eqn:plant_dyn} under the feedback law $u
  = K \hat{x}$, with $t\mapsto \hat{x}(t)$ evolving according
  to~\eqref{eqn:x_hat} and the sequence $\{t_k\}_{k \in
    \integersnonnegative}$ determined recursively by
  \begin{align}\label{eqn:tk_trig3}
    t_{k+1} = \min \{ t \geq r_k: \ &\tilde{b}(T_M, b(t), \epsilon(t))
    \geq 1 \ \operatorname{OR}
    \\
    & \triggerChbar( T_M, b(t), \epsilon(t), ( \epsilon(t) /
    2^{\bar{p}} ) ) \geq 1 \} , \notag
  \end{align}
  where $n \bar{p}$ is the upper bound on the number of bits that can
  be sent per transmission, $T > 0$ in the
  definition~\eqref{eqn:hbar_def} of $\triggerChbar$ is a design
  parameter, and $T_M$ is as given in
  Assumption~\eqref{A:transmit_time}.  Let $\{r_k\}_{k \in
    \integersnonnegative}$ be given as $r_0 = t_0$ and $r_k = t_k +
  \Delta_k$ for $k \in \integerspositive$.  Assume the encoding scheme
  is such that~\eqref{eqn:coding_scheme} is satisfied for all $t \ge t_0$.
  Further assume that $V(x(t_0)) \leq V_d(t_0)$, $\triggerChbar(T_M,
  b(t_0), \epsilon(t_0), ( \epsilon(t_0) / 2^{\bar{p}} ) ) \leq 1$ and
  that~\eqref{eq:W}-\eqref{eqn:V0_condition} hold.  Let $\uline{p_k}$
  be given by
  \begin{equation}\label{eqn:pk_lbound_case3}
    \uline{p_k} \! \triangleq \! \min \{ p \in \integerspositive : 
    \triggerChbar( T_M, b(t_k), \epsilon(t_k), \tfrac{\epsilon(t_k)}{2^p }
    ) \leq 1 \}.
  \end{equation}
  Then, the following hold:
  \begin{enumerate}
  \item $\uline{p_1} \leq \bar{p}$. Further for each $k \in
    \integerspositive$, if $ p_k \in \integerspositive \intersect [
    \uline{p_k}, \bar{p} ]$, then $\uline{p_{k+1}} \leq \bar{p}$.

  \item the inter-transmission times $\{ T_k = t_{k+1} - t_k \}_{k \in
      \integerspositive}$ and inter-reception times $\{ R_k \triangleq
    r_{k+1} - r_k \}_{k \in \integerspositive}$ have a uniform
    positive lower bound,

  \item the origin is exponentially practically stable for the
    closed-loop system, with $V(x(t)) \leq V_d(t) = ( V_d(t_0) - V_0)
    e^{ -\beta (t - t_0)} + V_0$ for all $t \geq t_0$.
  \end{enumerate}
\end{theorem}
\begin{proof}
  Since $V(x(t_0)) \leq V_d(t_0)$, i.e. $b(t_0) \leq 1$, and $
  \triggerChbar(T_M, b(t_0), \epsilon(t_0), ( \epsilon(t_0) /
  2^{\bar{p}} ) ) \leq 1$ the trigger~\eqref{eqn:tk_trig3} implies
  that $\uline{p_1} \leq \bar{p}$. Similarly, if for each $k \in
  \integerspositive$, $p_k \in \integerspositive \intersect [
  \uline{p_k}, \bar{p} ]$, then~\eqref{eqn:tk_trig3} implies
  $\uline{p_{k+1}} \leq \bar{p}$, which proves \emph{(i)}.
	
  Regarding \emph{(ii)}, note that Assumption~\eqref{A:transmit_time}
  implies that $r_k - t_k \geq 0$ for $k \in
  \integerspositive$. Therefore, it is enough to show that there
  exists a uniform lower bound on $t_{k+1} - r_k$. Notice 
  that~\eqref{eqn:pk_lbound_case3} implies that
  \begin{align*}
  \triggerChbar( T_M, b(t_k^-), \epsilon(t_k^-), (
  \epsilon(t_k^-) / 2^{ p_k } ) ) \leq 1,
  \end{align*}
  which in turn implies, as a consequence of the fact that $\Delta_k 
  \leq T_M$ and Lemma~\ref{lem:Gamma2-T_sign}, that 
  $\tilde{\Gamma}_2( b(t_k^-), \epsilon(t_k^-), ( \epsilon(t_k^-) / 
  2^{ p_k } ) ) - \Delta_k \geq 0$. Invoking 
  Lemma~\ref{lem:Gamma2-T_sign} once more, we see that
  \begin{align*}
    \triggerCh(r_k) &\leq \triggerChbar( \Delta_k, b(t_k^-),
    \epsilon(t_k^-), ( \epsilon(t_k^-) / 2^{p_k} ) ) \leq 1.
  \end{align*}
  In other words, $\epsilon(r_k) \leq \rhofun{T}{ b(r_k) }$. Now,
  let us pick $\tilde{T} \in (0, T)$ and notice that
  Lemma~\ref{lem:Gamma1_prop} guarantees that for all $\epsilon_0 \leq
  \rhofun{\tilde{T}}{ b_0 }$, $\Gamma_1(b_0, \epsilon_0) \geq \min \{
  \Gamma_1(1, 1), \tilde{T} \}$. Since $\tilde{T} \in (0, T)$, there
  exists a constant $\varpi \in (0, 1)$ such that $\epsilon(r_k) \leq
  \varpi \rhofun{ \tilde{T} }{ b(r_k) }$. Thus, again for all
  $\epsilon_0 \leq \rhofun{\tilde{T}}{ b_0 }$, we have that
  $\tilde{\Gamma}_2(b_0, \epsilon_0, \psi_0) \geq \min \{ \Gamma_1(1,
  1), \tilde{T}, T^\bullet \}$, with
  \begin{equation*}
    T^\bullet \triangleq \min \{ \tau \geq 0 : \phi \big( \tau, 
    \varpi / ( 2^{\bar{p}} ) \big) = 1 \}.
  \end{equation*}
  Since $T_M < T$ by Assumption~\eqref{A:transmit_time}, there exists
  a choice of $\tilde{T} \in (T_M , T)$ such that $T_M < \min \{
  \Gamma_1(1, 1), \tilde{T}, T^\bullet \}$.  Thus, by
  Lemma~\ref{lem:Gamma2-T_sign}, we have that for all $\epsilon_0 \leq
  \rhofun{\tilde{T}}{ b_0 }$, $\triggerChbar( T_M, b_0, \epsilon_0, (
  \epsilon_0 / 2^{\bar{p}} ) ) < 1$. As a consequence, for $k \in
  \integersnonnegative$, $t_{k+1} - r_k$ is uniformly lower bounded by
  the time it takes $\displaystyle \frac{ \epsilon(t) }{ \rhofun{
      \tilde{T} }{ b(t) } }$ to evolve from $\varpi$ to $1$, which in
  turn can be shown to have a uniform positive lower bound following
  arguments analogous to those in the proof of
  Theorem~\ref{thm:inst_tx_bounded_bit}.
 
  Regarding \emph{(iii)}, note that from the triggering
  rule~\eqref{eqn:tk_trig3}, we see that $\tilde{b}(T_M, b(t_k),
  \epsilon(t_k)) \leq 1$, which from Lemma~\ref{lem:Gamma1-T_sign}
  implies that $\Gamma_1(b(t_k), \epsilon(t_k)) \geq T_M$. In other
  words, $V(x(t)) \leq V_d(t)$ (i.e., $b(t) \leq 1$) for \textit{at
    least} all $t \in [t_k, r_k]$ for any $k \in
  \integersnonnegative$. Since $\triggerChbar(T_M, b(t_0),
  \epsilon(t_0), ( \epsilon(t_0) / 2^{\bar{p}} ) ) \leq 1$ it means
  that $\epsilon(r_0) \leq \rhofun{T}{ b(r_0) }$.  Further, we have
  already seen that for any $k \in \integerspositive$, $\epsilon(r_k)
  \leq \rhofun{T}{ b(r_k) }$. Therefore, for any $k \in
  \integersnonnegative$, $\Gamma_1(b(r_k), \epsilon(r_k)) \geq
  \Gamma_1(1,1) > 0$. This means $V(x(t)) \leq V_d(t)$ (i.e., $b(t)
  \leq 1$) for \textit{at least} all $t \in [r_k, t_{k+1}]$. Putting
  these two facts together with \emph{(ii)} concludes the proof.
\end{proof}

Despite its appearance, note that the event-triggering
rule~\eqref{eqn:tk_trig3} in Theorem~\ref{thm:noninst_com_bounded_bit}
is a generalization of the rule~\eqref{eqn:tk_trig2} in
Theorem~\ref{thm:inst_tx_bounded_bit}. In fact, when communication is
instantaneous, $T_M = 0$, and we have $\tilde{b}(T_M, b(t),
\epsilon(t)) = b(t)$ and $\triggerChbar( T_M, b(t), \epsilon(t), (
\epsilon(t) / 2^{\bar{p}} ) ) = \triggerCh(t) / ( 2^{\bar{p}} )$.

\begin{remark}\longthmtitle{Tuning the parameter $T$}
  {\rm The parameter $T$ in~\eqref{eqn:rho_def} presents a trade-off
    between maximum allowable communication delay $T_M$ and
    inter-transmission times through $\epsilon$ (in the sense of
    $\tilde{b}(T_M, b(t), \epsilon(t)) \leq 1$). The smaller the value
    of $T$, the greater the tolerable $\epsilon$ and the
    inter-transmission times are, at the cost of a potentially smaller
    $T_M$.} \oprocend
\end{remark}

We let $\bth = \Infnorm{A} + \frac{\beta}{2}$ in the sequel. The next
result upper bounds $\uline{p_k}$ in terms of the history of the
number of bits transmitted.

\begin{corollary}\longthmtitle{Upper Bound on $\uline{p_k}$ in Terms
    of the History of the Number of Bits
    Transmitted}\label{cor:sufficient-bit-rate-general}
  Under the assumptions of Theorem~\ref{thm:noninst_com_bounded_bit},
  the following holds for any $k \in \integerspositive$,
  \begin{align*}
    \uline{p_k} &\leq \log_2 \bigg( \frac{ e^{\bth T_M} }{
      \rhofun{T}{\tilde{b}(T_M, b(t_k^-), \epsilon(t_k^-)} -
      \alpha(T_M) } \bigg) + 1
    \\
    &\quad + \log_2 \bigg( \frac{ e^{\bth (t_k - t_0) } }{
      \prod_{j=1}^{k-1} 2^{ p_j } } \epsilon(t_0) + \sum_{i=0}^{k-1}
    \prod_{j=i+1}^{k-1} \frac{ e^{\bth T_j} }{2^{p_j}} \alpha(T_i)
    \bigg),
  \end{align*}
  with $\displaystyle \alpha(\tau) \triangleq \frac{ \nu (
    e^{\Enorm{A} \tau} - 1 ) }{ c \Enorm{A} \sqrt{ V_0 } }$.
\end{corollary}
\begin{proof}
  Using~\eqref{eqn:eps_btild_def} and~\eqref{eqn:coding_scheme}
  recursively along with the fact that $V_d(t) \geq V_0$ for all $t
  \geq t_0$ gives for $k \in \integerspositive$
  \begin{align}\label{eqn:eps_bound}
    \epsilon(t_k^-) &\leq e^{\bth T_{k-1}} \frac{
      \epsilon(t_{k-1}^-) }{ 2^{p_{k-1}}} + \alpha(T_{k-1}) \notag
    \\
    &\leq \frac{ e^{\bth T_{k-1}} }{ 2^{p_{k-1}}} \Big[ \frac{
      e^{\bth T_{k-2}} }{ 2^{p_{k-2}}} \epsilon(t_{k-2}^-) +
    \alpha(T_{k-2}) \Big] + \alpha(T_{k-1}) \notag
    \\
    &\leq \frac{ e^{\bth (t_k - t_0) } }{ \prod_{j=1}^{k-1} 2^{ p_j
      } } \epsilon(t_0) + \sum_{i=0}^{k-1} \prod_{j=i+1}^{k-1} \frac{
      e^{\bth T_j} }{2^{p_j}} \alpha(T_i).
  \end{align}
  Next, observe that, for each $k \in \integerspositive$,
  $\epsilon(t_k^-) \geq \rhofun{T}{b(t_k^-)}$. If this were not the
  case, then $\Gamma_1( b(t_k^-), \epsilon(t_k^-) ) \geq \min \{
  \Gamma_1(1, 1), T \}$ by Lemma~\ref{lem:Gamma1_prop}, and on the
  other hand $\tilde{\Gamma}_2(b(t_k^-), \epsilon(t_k^-),
  \epsilon(t_k^-)/2^{\bar{p}}) \geq \min \{ \Gamma_1(1, 1), T, T^* \}
  > T_M$. These two conditions together would imply, by
  Lemmas~\ref{lem:Gamma1-T_sign} and~\ref{lem:Gamma2-T_sign}, that
  neither of the conditions in the trigger~\eqref{eqn:tk_trig3} is
  satisfied at $t = t_k^-$, which is a contradiction.

  Now, since Theorem~\ref{thm:noninst_com_bounded_bit} guarantees
  $b(t) \leq 1$ for all $t \geq t_0$ and since $\rhofun{T}{y} \geq 1$
  for all $y \in [0, 1]$, we have $\epsilon(t_k^-) \geq 1$.  Next, the
  trigger~\eqref{eqn:tk_trig3} and
  Theorem~\ref{thm:noninst_com_bounded_bit}(i) ensure that
  $\triggerChbar( T_M, b(t_k^-), \epsilon(t_k^-), ( \epsilon(t_k^-) /
  2^{\bar{p}} ) ) \leq 1$, i.e.,
  \begin{equation*}
    \Infnorm{e^{A T_M}} e^{(\beta/2) T_M} \frac{ \epsilon(t_k^-) }{ 2^{\bar{p}} } + 
    \alpha(T_M) \leq \rhofun{T}{\tilde{b}(T_M, b(t_k^-), 
      \epsilon(t_k^-)}.
  \end{equation*}
  Rearranging the terms and using the fact $\epsilon(t_k^-) \geq 1$,
  we have
  \begin{equation*}
    \rhofun{T}{\tilde{b}(T_M, b(t_k^-), \epsilon(t_k^-)} - 
    \alpha(T_M) \geq \frac{ \Infnorm{e^{A T_M}} e^{(\beta/2) T_M} }{
      2^{\bar{p}} } > 0. 
  \end{equation*}
  Now~\eqref{eqn:pk_lbound_case3} and the fact $e^{\Infnorm{A}T_M}
  \geq \Infnorm{e^{AT_M}}$ imply that
  \begin{equation*}
    e^{\bth T_M} \frac{ \epsilon(t_k^-) }{ 2^{( \uline{p_k} - 1 )} 
    } + \alpha(T_M) \geq \rhofun{T}{\tilde{b}(T_M, b(t_k^-), 
      \epsilon(t_k^-)} ,
  \end{equation*}
  which in turn gives
  \begin{equation*}
    2^{\uline{p_k}} \leq \frac{ 2 e^{\bth T_M} \epsilon(t_k^-) }{ 
      \rhofun{T}{\tilde{b}(T_M, b(t_k^-), \epsilon(t_k^-)} - \alpha(T_M) 
    }.
  \end{equation*}
  In other words,
  \begin{equation*}
    \uline{p_k} \leq \log_2( \epsilon(t_k^-) ) + \log_2 \bigg( \frac{ 2
      e^{\bth T_M} }{ \rhofun{T}{\tilde{b}(T_M, 
        b(t_k^-), \epsilon(t_k^-)} - \alpha(T_M) } \bigg).
  \end{equation*}
  Substituting~\eqref{eqn:eps_bound} (and multiplying by $n$ to give
  us the number of bits) yields the result.
\end{proof}

Although this result does not explicitly give a data rate as in
Corollary~\ref{cor:sufficient-bit-rate}, it provides an implicit
characterization of it. This becomes more clear in the absence of
disturbances.

\begin{corollary}\longthmtitle{Upper Bound on Sufficient Data Rate in
    the Absence of
    Disturbances}\label{cor:sufficient_bit_rate_zero_dist}
  Under the assumptions of Theorem~\ref{thm:noninst_com_bounded_bit}
  and no disturbance, the following holds for any $k \in
  \integerspositive$,
    \begin{align*}
      n \Big( \uline{p_k} + \sum_{i=1}^{k-1} p_i \Big) & \leq n \bigg[
      \log_2 \bigg( \frac{ e^{\bth T_M} }{ \rhofun{T}{\tilde{b}(T_M,
          b(t_k^-), \epsilon(t_k^-)} } \bigg)
      \\
      & \quad + 1 + \bth \log_2 ( e ) (t_k - t_0) + \log_2(
      \epsilon(t_0) ) \bigg].
    \end{align*}
\end{corollary}
\begin{proof}
  In the no disturbance case, $\nu = 0$ and the second term
  of~\eqref{eqn:de_evolve_rk} vanishes, which justifies $\alpha(\tau)
  \equiv 0$ in Corollary~\ref{cor:sufficient-bit-rate-general} even in
  the case $V_0 = 0$. As a result, we have
  \begin{align*}
    \uline{p_k} &\leq \log_2 \bigg( \frac{ e^{\bth T_M} }{ 
      \rhofun{T}{\tilde{b}(T_M, b(t_k^-), \epsilon(t_k^-)} } \bigg) + 1
    \\
    &\quad + \log_2 \bigg( \frac{ e^{\bth (t_k - t_0) } }{ 
      \prod_{j=1}^{k-1} 2^{ p_j } } \epsilon(t_0) \bigg).
  \end{align*}
  Multiplying by $n$ and rearranging the terms yields the sufficient
  data rate in the statement.
\end{proof}

Note that the effect of non-instantaneous communication, through
$T_M$, in Corollary~\ref{cor:sufficient_bit_rate_zero_dist} only has a
transient effect on the sufficient data rate. If $T_M = 0$, the first
term is non-positive (recall $\rho_T \geq 1$) and we recover the
result of Corollary~\ref{cor:sufficient-bit-rate}.

\section{Simulations}\label{sec:sim}

We illustrate our results in simulation for three scenarios:
instantaneous communication with no disturbance and non-instantaneous
communication with and without
disturbance. % We begin by describing the
% problem data.
Consider the system on $\real^2$ given
by~\eqref{eqn:plant_dyn} with
\begin{equation*}
  A = 
  \begin{bmatrix}
    1 & -2
    \\
    1 & 4
  \end{bmatrix}, \ 
  B = 
  \begin{bmatrix}
    0
    \\
    1
  \end{bmatrix}, \ K =
  \begin{bmatrix}
    2 & -8
  \end{bmatrix} .
\end{equation*}
The plant matrix $A$ has eigenvalues at $2$ and $3$, while the control
gain matrix $K$ places the eigenvalues of the matrix $\bar{A} = A+BK$
at $-1$ and $-2$. We select the matrix $Q=\identity{2}$, for which the
solution to the Lyapunov equation~\eqref{eq:Lyap-eq} is
\begin{align*}
  P =
  \begin{bmatrix}
    2.2500 & -0.9167
    \\
    -0.9167 & 0.5833
  \end{bmatrix}
  .
\end{align*}
The desired control performance is specified by
\begin{align*}
  V_d(t_0) = 1.1 V(x(t_0)), \quad \beta = 0.8
  \frac{\lambda_m(Q)}{\lambda_M(P)} ,
\end{align*}
and $V_0$ chosen according to~\eqref{eqn:V0_condition} in each
scenario. We set $a = 1.2$ in~\eqref{eq:W}, so that $W > 0$, and
assume, without loss of generality, $t_0 = 0$. We choose the design
parameter $T = 0.5 \times \Gamma_1(1, 1)$. The initial condition is 
$x(t_0) = ( 6, -4)$, and the encoder and decoder use the information
\begin{align*}
  \hat{x}(t_0) = ( 0, 0) , \quad d_e(t_0) = 2 \Infnorm{x(t_0) -
    \hat{x}(t_0)} .
\end{align*}
Finally, unless specified otherwise, the number of bits transmitted at
each transmission time is $n \uline{p_k}$, the ``minimum" number of bits
as prescribed by~\eqref{eqn:pk_lbound_case2}
and~\eqref{eqn:pk_lbound_case3}, respectively.

\textit{Instantaneous communication and no disturbance:} we let $\nu =
V_0 = 0$, for which we obtain $\Gamma_1(1,1) = 0.5699$. We present
simulations for two cases, $\bar{p} = 12$ and $\bar{p} = 20$, where $n
\bar{p} = 2\bar{p}$ is the uniform upper bound on the number of bits
per transmission imposed by the communication channel.
\begin{figure}[!htpb]
  \centering
  \subfigure[\label{fig:pbar_12_V}]{\includegraphics[width=0.24\textwidth,height=1.3
    in]{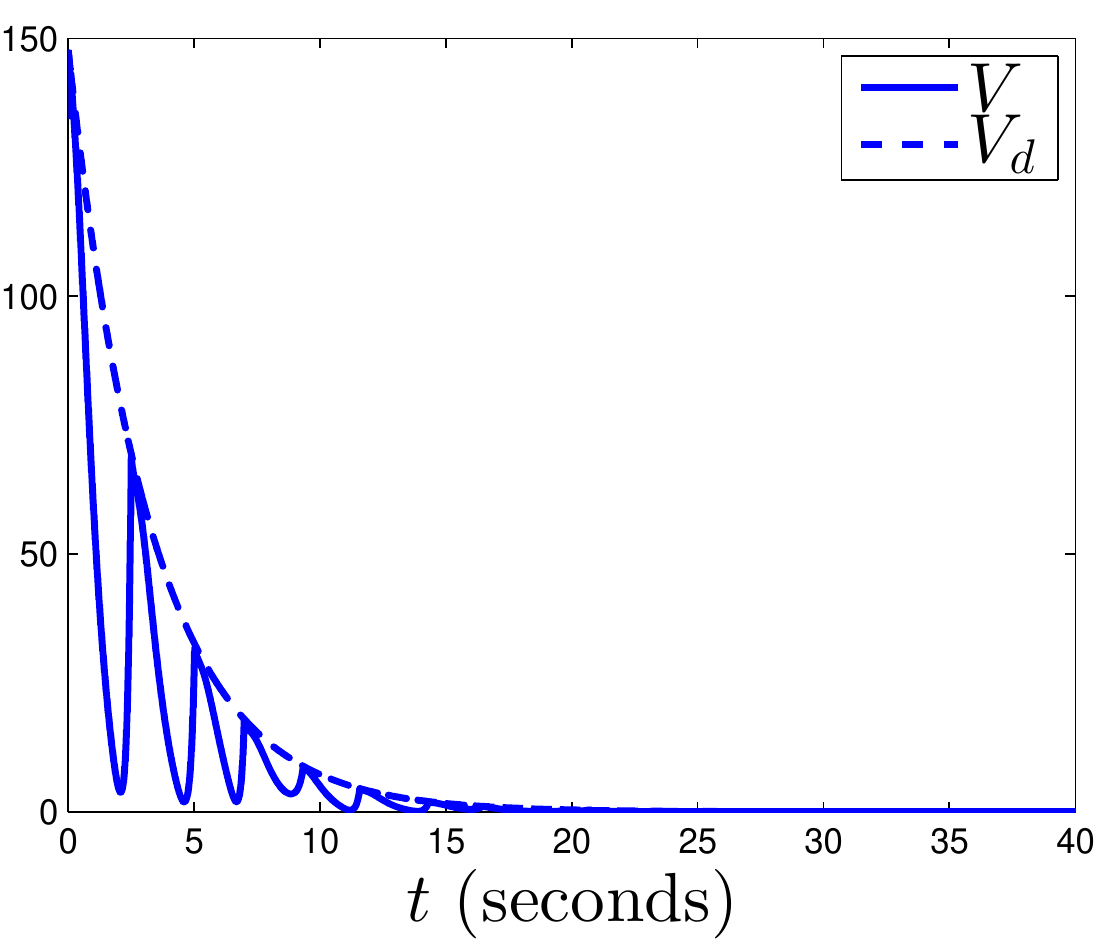}}
  \subfigure[\label{fig:pbar_20_V}]{\includegraphics[width=0.24\textwidth,height=1.3
    in]{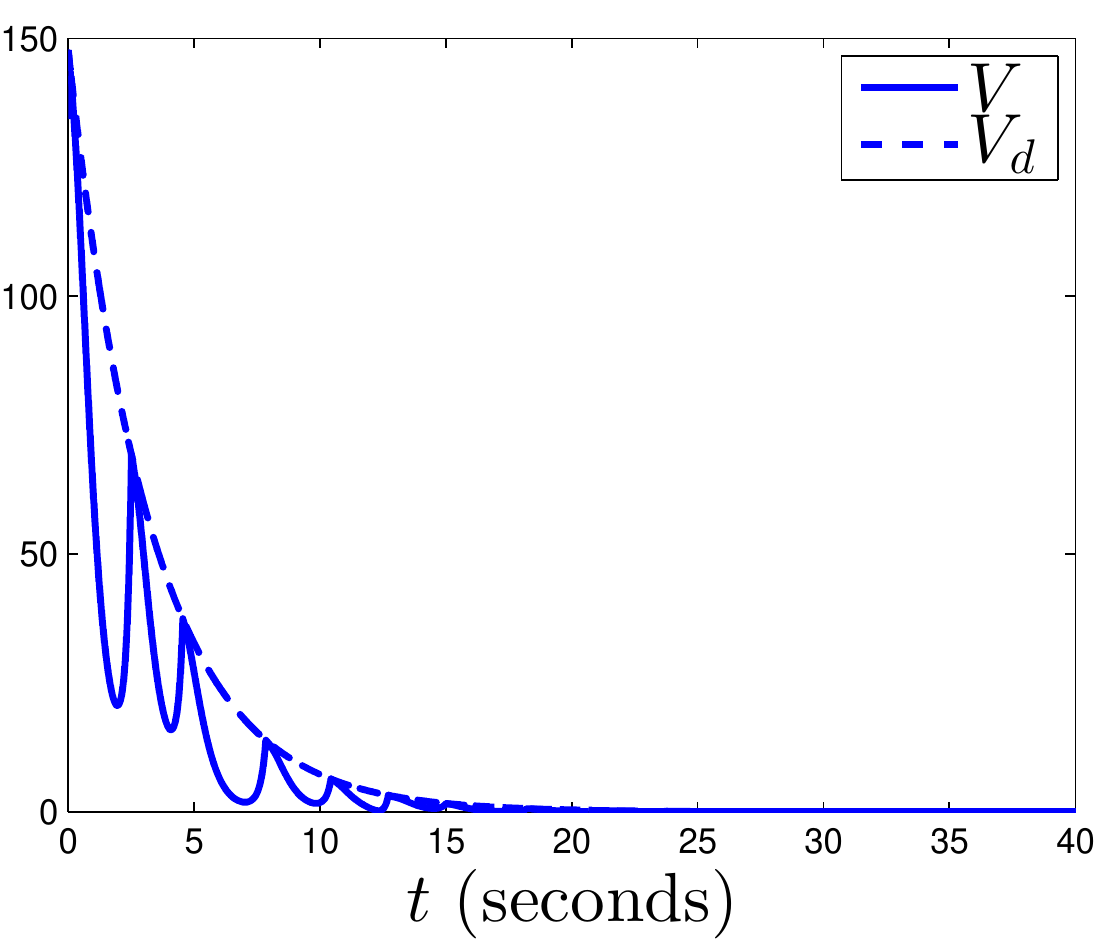}}
  \vspace*{-1ex}
  \caption{Instantaneous communication and no disturbance: evolution
    of $V_d$ and $V$ under the event-triggered
    design~\eqref{eqn:tk_trig2} for (a) $\bar{p} = 12$ and (b)
    $\bar{p} = 20$.}\label{fig:V_Vd}
\end{figure}

Figure~\ref{fig:V_Vd} shows the evolution of $V$ and $V_d$ in both
cases. As established in Theorem~\ref{thm:inst_tx_bounded_bit}, the
desired convergence rate is guaranteed in each case. In the case of
$\bar{p} = 12$, it turns out that $\uline{p_k} = \bar{p}$ for each $k
\in \integerspositive$. On the other hand, in the case when $\bar{p} =
20$, the performance of $V$ with respect to $V_d$ plays a more
relevant role in determining the transmission times
in~\eqref{eqn:tk_trig2}. In fact, in the presented simulation, $
\uline{p_k} < \bar{p}$ on all transmissions, as depicted in
Figure~\ref{fig:bit_profile}. Figure~\ref{fig:bit_rate_compare} shows
the interpolated plot of the total number of bits transmitted for both
cases, $\bar{p} = 12$ and $\bar{p} = 20$ along with sufficient
(Corollary~\ref{cor:sufficient-bit-rate}) and necessary
(Proposition~\ref{prop:necs_rate}) data rate.
\begin{figure}[!htpb]
  \centering
  \subfigure[\label{fig:bit_profile}]{\includegraphics[width=0.23\textwidth,height=1.2
    in]{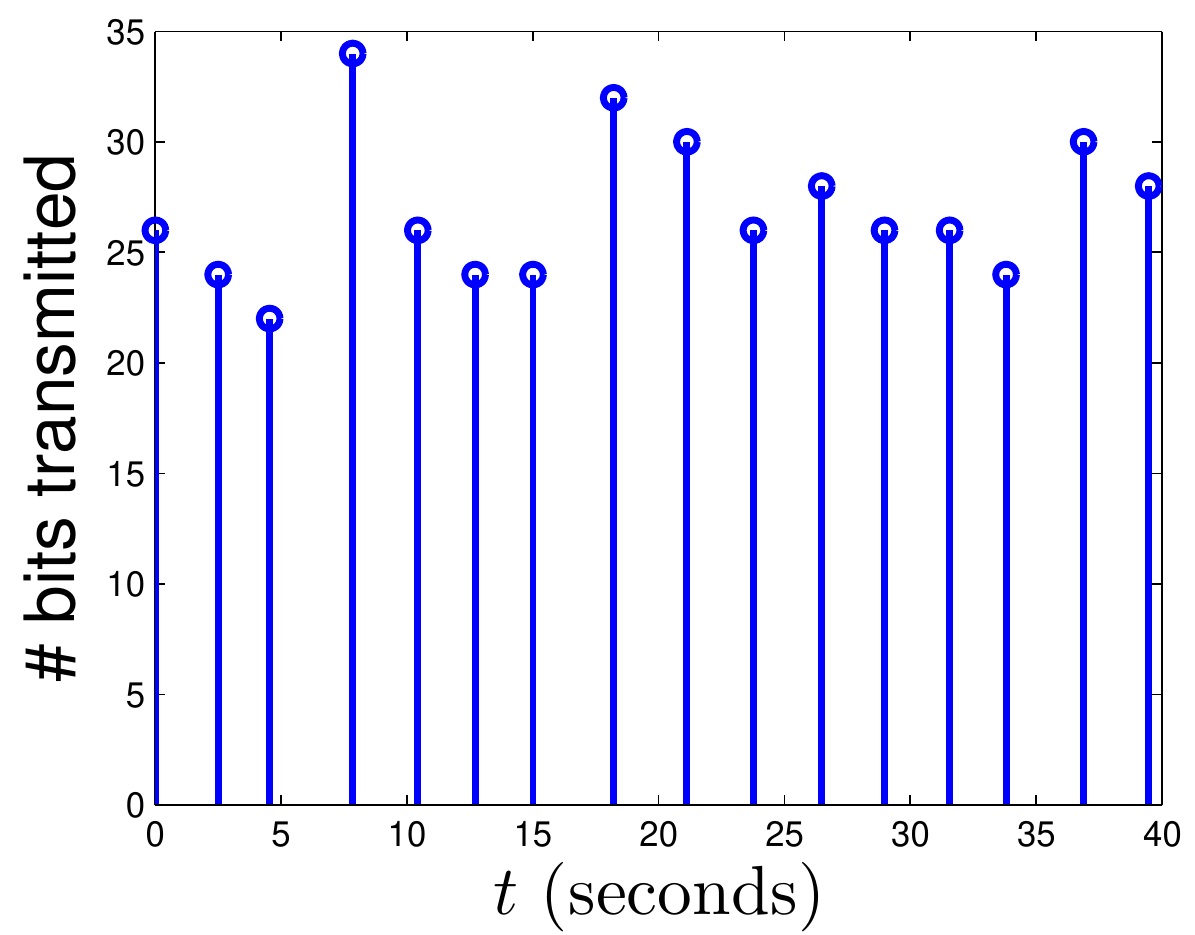}}
  \subfigure[\label{fig:bit_rate_compare}]{\includegraphics[width=0.23\textwidth,height=1.2
    in]{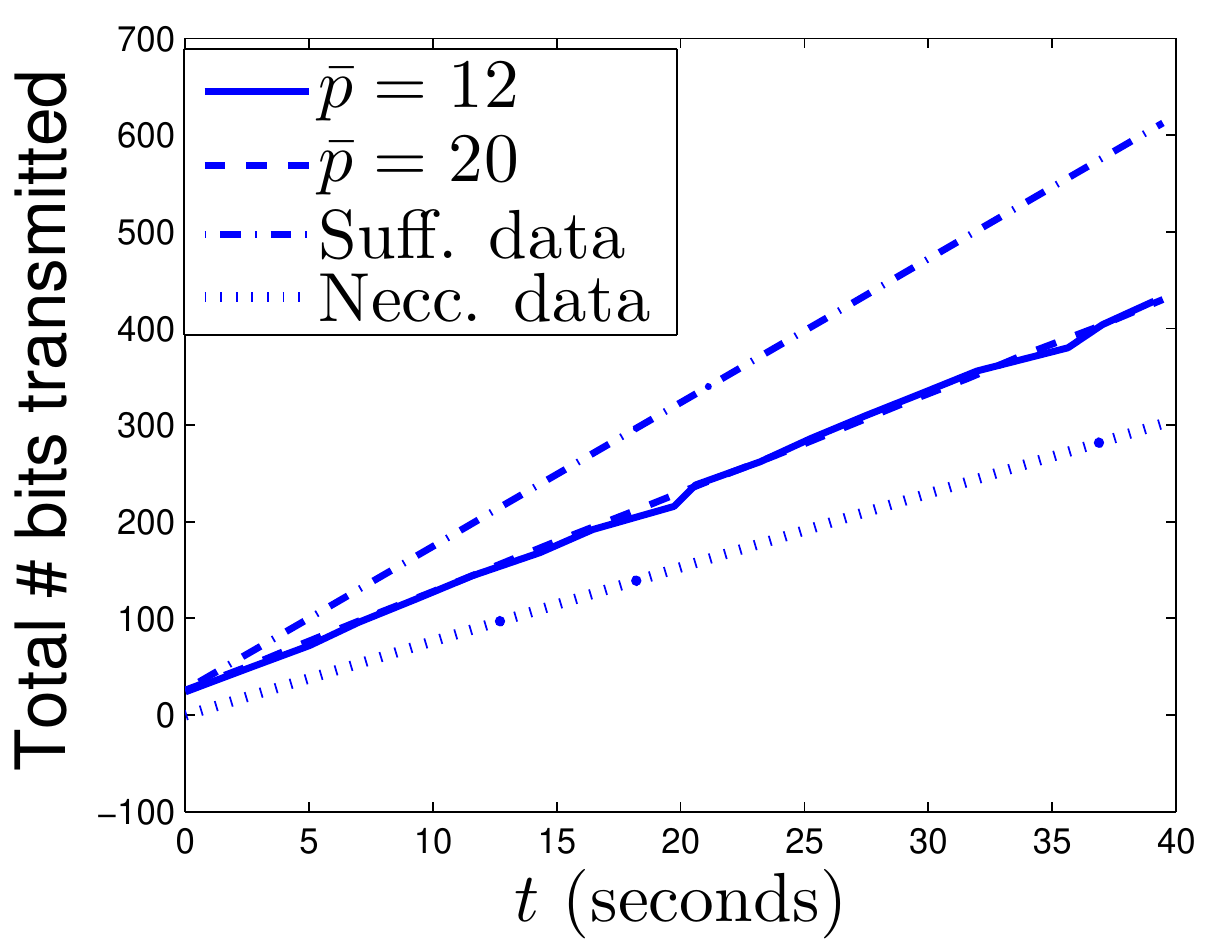}}
  \vspace*{-1ex}
  \caption{For the event-triggered implementations in
    Figure~\ref{fig:V_Vd}, (a) shows the number of bits on each
    transmission when $\bar{p} = 20$ and (b) shows the interpolated
    plot of the total number of bits transmitted when $\bar{p} = 12,
    20$; along with sufficient
    (Corollary~\ref{cor:sufficient_bit_rate_zero_dist}) and necessary
    (Proposition~\ref{prop:necs_rate}) data.} \label{fig:bit_evolve}
\end{figure}
Although in reality the total number of bits transmitted as a function
of time is piecewise constant, the interpolated plots enable a more
insightful comparison. In the case $\bar{p} = 20$, after having
transmitted more bits initially than for $\bar{p} = 12$, the gap in
the cumulative bit counts diminishes eventually. Finally, during the
time interval $[0,40]$, the number of transmissions, average
inter-transmission time, and minimum inter-transmission time are $18$,
$2.2995$ and $0.8585$ (case $\bar{p} = 12$) and $16$, $2.63$ and
$2.048$ (case $\bar{p} = 20$), respectively.

\textit{Non-instantaneous communication and non-zero disturbance:} we
let $\nu = 0.01$ and, following~\eqref{eqn:V0_condition} with $\sigma
= 0.9$, we set $V_0 = 5.3942$, for which we obtain $\Gamma_1(1,1) =
0.0347$. The actual disturbance signal employed in the simulation is
\begin{equation*}
  v_1(t) = \nu \sin(0.5 t), \quad v_2(t) = \nu \cos(0.5 t).
\end{equation*}
We present a simulation for the case $\bar{p} = 20$. We choose $T_M =
0.8 \times \min \{ \Gamma_1(1, 1), T, T^* \} = 1.1 \times 10^{-3}$ and
the communication time $\Delta_k = r_k - t_k=T_M$ for all $k \in
\integerspositive$ (consequently, note that $R_k = T_k$ for $k \in
\integersnonnegative$).
\begin{figure}[!htpb]
  \centering \subfigure[\label{fig:non_inst_pbar_20_V}]
  {\includegraphics[width=0.24\textwidth,height=1.3in]
    {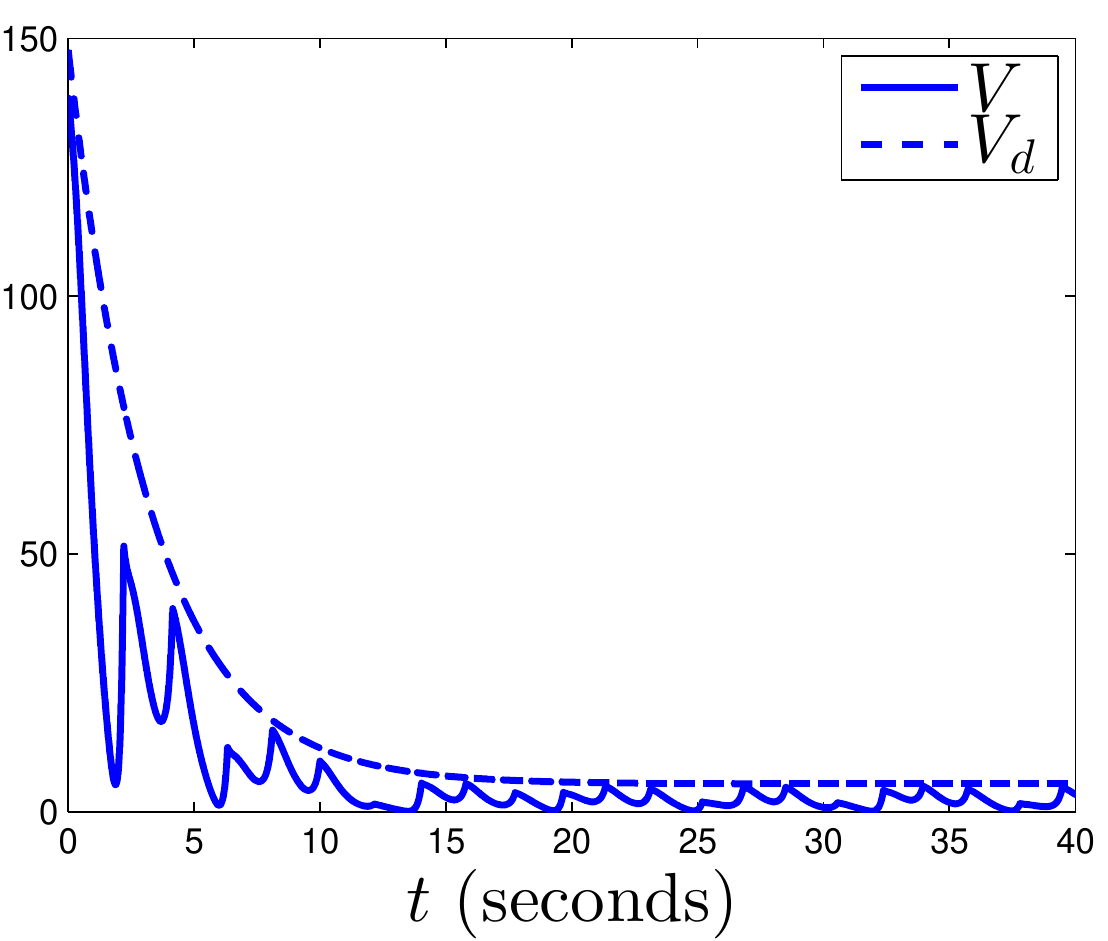}}
  \subfigure[\label{fig:non_inst_pbar_20_inter-tx_times}]
  {\includegraphics[width=0.24\textwidth,height=1.3in]
    {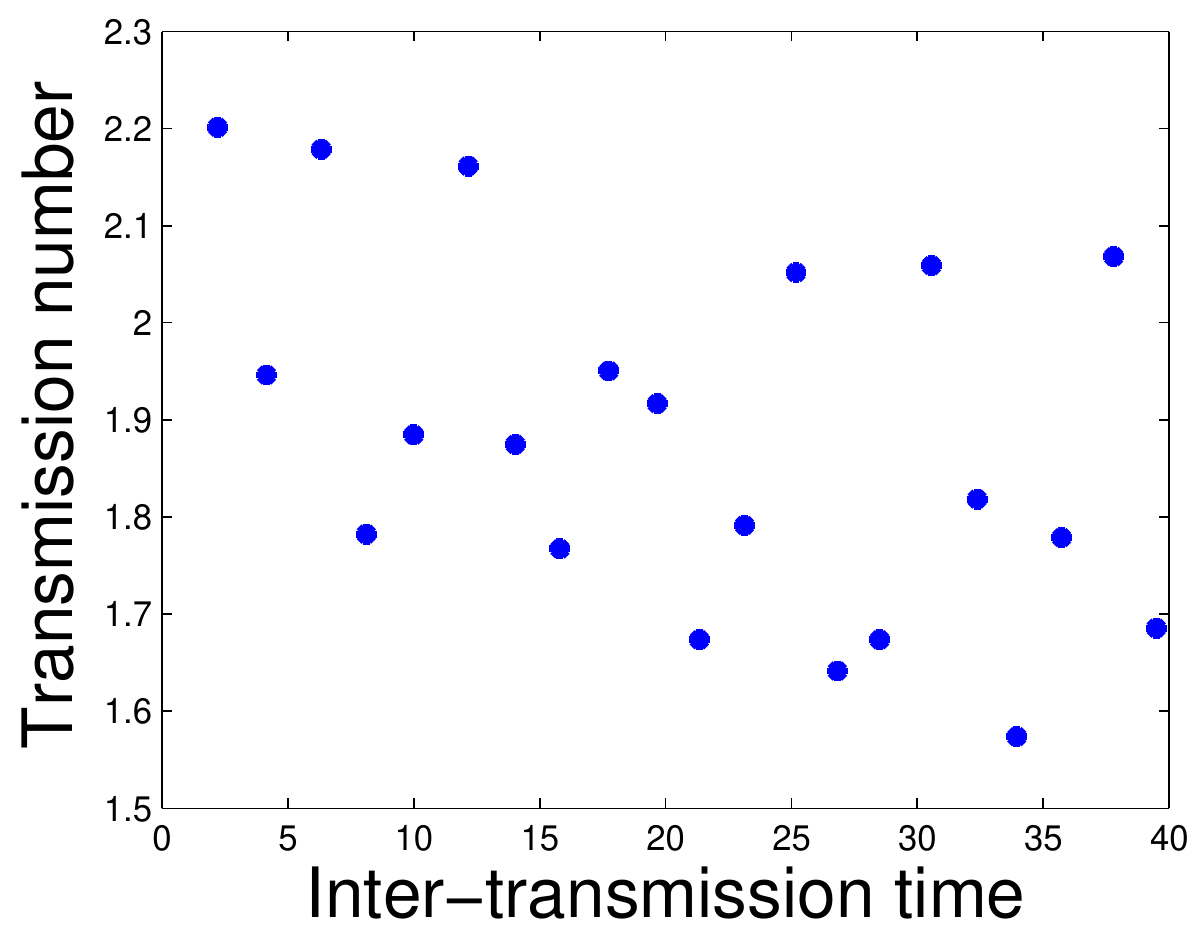}}
  \vspace*{-1ex}
\caption{Non-instantaneous communication and non-zero disturbance: (a)
  shows the evolution of $V_d$ and $V$ and (b) shows the
  inter-transmission times under the event-triggered
  design~\eqref{eqn:tk_trig3} for $\bar{p} =
  20$.}\label{fig:non_inst_pbar_20}
\end{figure}
Figure~\ref{fig:non_inst_pbar_20_V} shows the evolution of $V$ and
$V_d$, which is in accordance with
Theorem~\ref{thm:noninst_com_bounded_bit}, while
Figure~\ref{fig:non_inst_pbar_20_inter-tx_times} shows the
inter-transmission times.
\begin{figure}[!htpb]
  \centering
  \subfigure[\label{fig:non_inst_bit_profile}]
  {\includegraphics[width=0.23\textwidth,height=1.2in] 
  {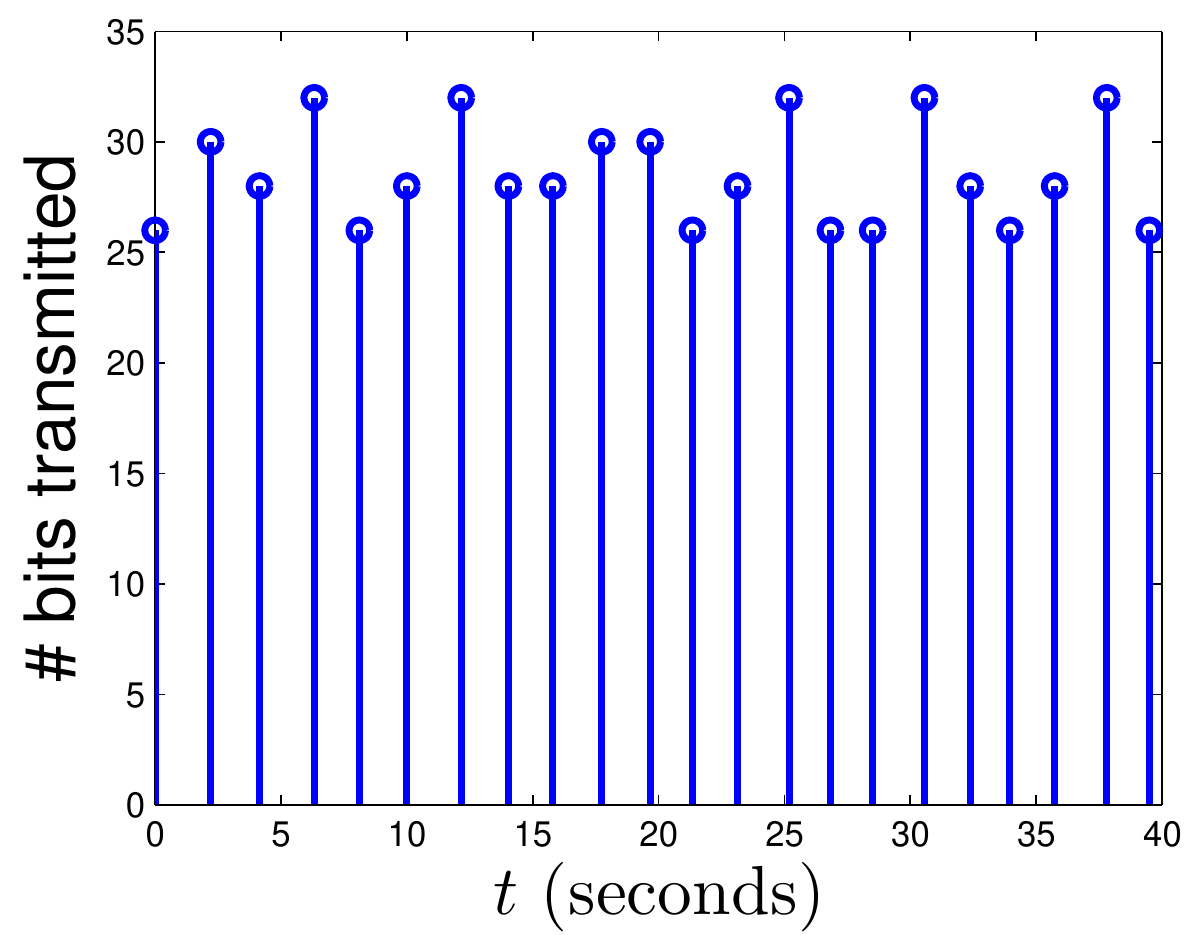}}
  \subfigure[\label{fig:non_inst_pbar_20_cum_bit}] 
  {\includegraphics[width=0.23\textwidth,height=1.2in] 
  {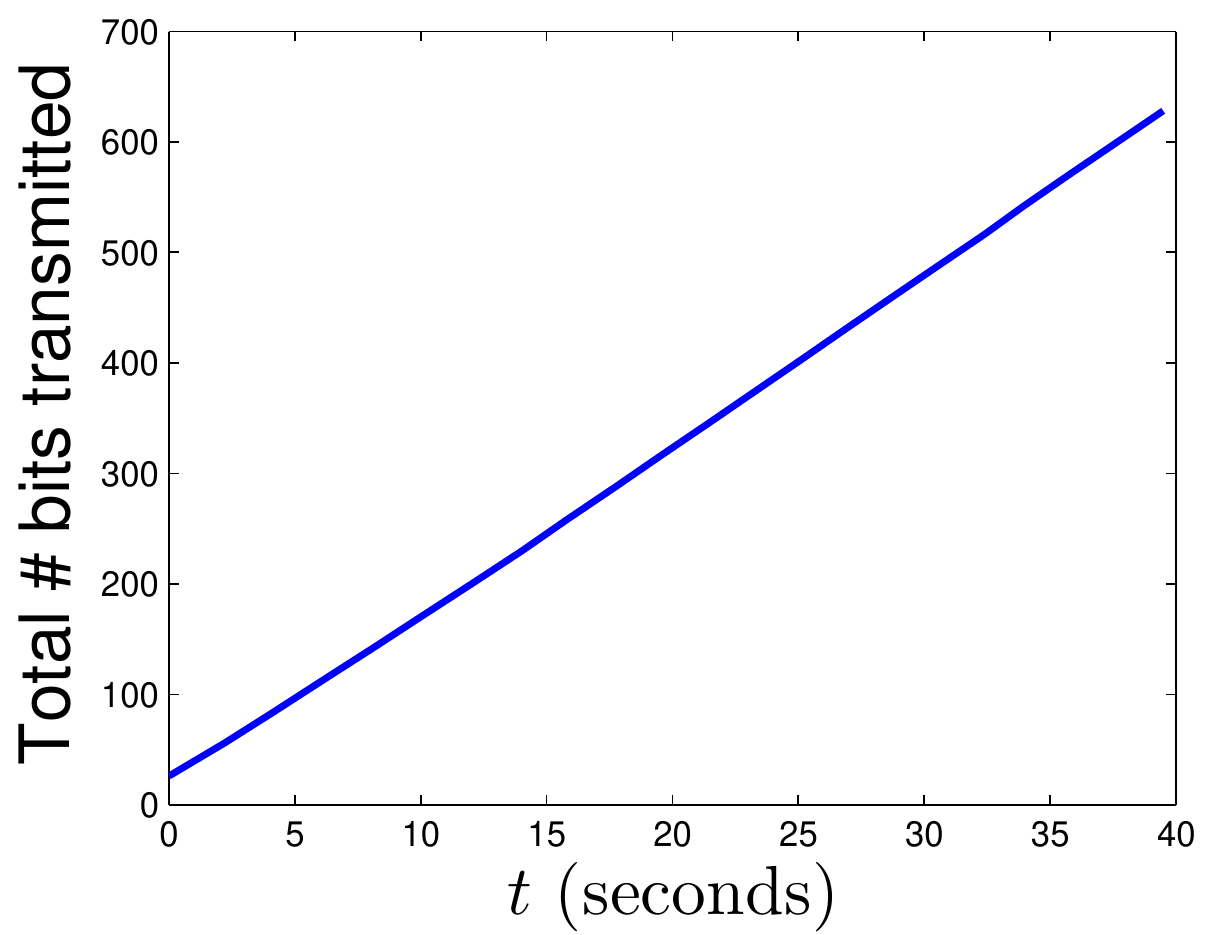}}
  \vspace*{-1ex}
  \caption{For the event-triggered implementation shown in
    Figure~\ref{fig:non_inst_pbar_20_V} with $\bar{p} = 20$, (a) shows
    the number of bits on each transmission and (b) shows the
    interpolated plot of the total number of bits
    transmitted.}\label{fig:non_inst_bit_evolve}
\end{figure}
Figure~\ref{fig:non_inst_bit_evolve} displays the evolution of the
number of bits transmitted. During the time interval $[0,40]$, the
number of transmissions is $22$, with average and minimum
inter-transmission intervals of $1.8182$ and $1.5739$, respectively.

\textit{Non-instantaneous communication and no disturbance:} we let
$\nu = V_0 = 0$, $\bar{p} = 20$ and $T_M = 1.1 \times 10^{-3}$. The
values of $\Gamma_1(1,1)$ and $T$ are as in the case of instantaneous
communication with no disturbance. We choose the communication time as
$\Delta_k = r_k - t_k=T_M$ for all $k \in \integerspositive$. To
illustrate Corollary~\ref{cor:sufficient_bit_rate_zero_dist}, we
compare the results of two simulations: in ``Sim1'' we choose $p_k =
\uline{p_k}$ for all $k \in \integerspositive$ while in ``Sim2'' we
choose $p_k = \bar{p}$ for $k \in \{1,2,3,4\}$ and $p_k = \uline{p_k}$
for all $k \in [5, \infty) \cap \integerspositive$.
Figure~\ref{fig:sim2_bit_profile} shows the number of bits on each
transmission for ``Sim2'' while
Figure~\ref{fig:bit_rate_compare_non_inst} compares the interpolated
total number of bits transmitted in ``Sim1'' and ``Sim2''. Notice that
until $5^{\text{th}}$ transmission time of ``Sim2'', the cumulative
bit count for ``Sim2'' exceeds that of ``Sim1'' but the gap is
immediately closed at that time and thereafter remains slightly lower
than that of ``Sim1''. This demonstrates the ability of the
event-trigger design to transmit fewer bits if more bits than
prescribed were transmitted in the past. We also see that the data
rate, as interpreted in
Corollary~\ref{cor:sufficient_bit_rate_zero_dist}, remains
approximately fixed irrespective of the past history of transmitted
bit count as long as the constraints of
Theorem~\ref{thm:noninst_com_bounded_bit} are respected. We have not
observed a similar behavior in the scenario with disturbance.
\begin{figure}[!htpb]
  \centering
  \subfigure[\label{fig:sim2_bit_profile}]
  {\includegraphics[width=0.23\textwidth,height=1.2in] 
    {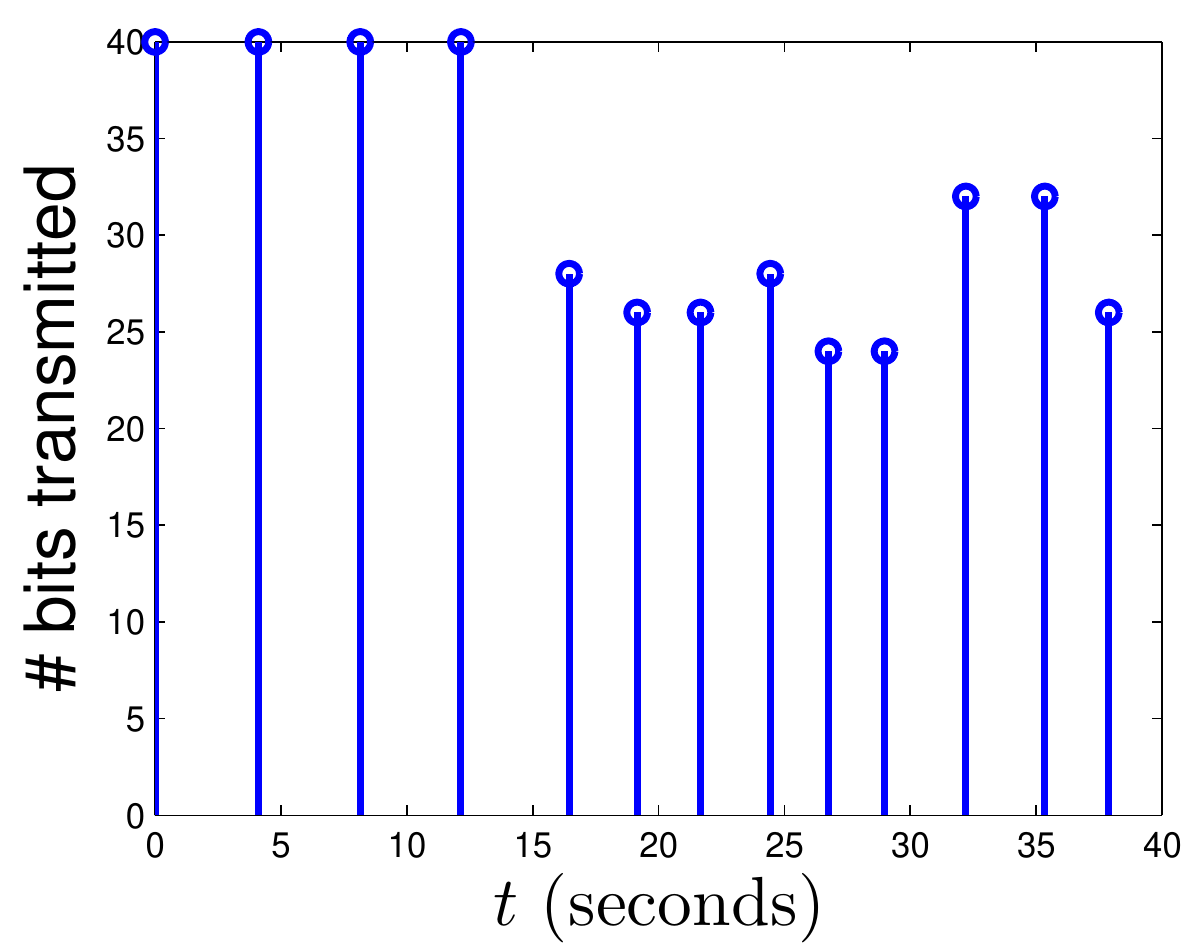}}
  \subfigure[\label{fig:bit_rate_compare_non_inst}] 
  {\includegraphics[width=0.23\textwidth,height=1.2in] 
    {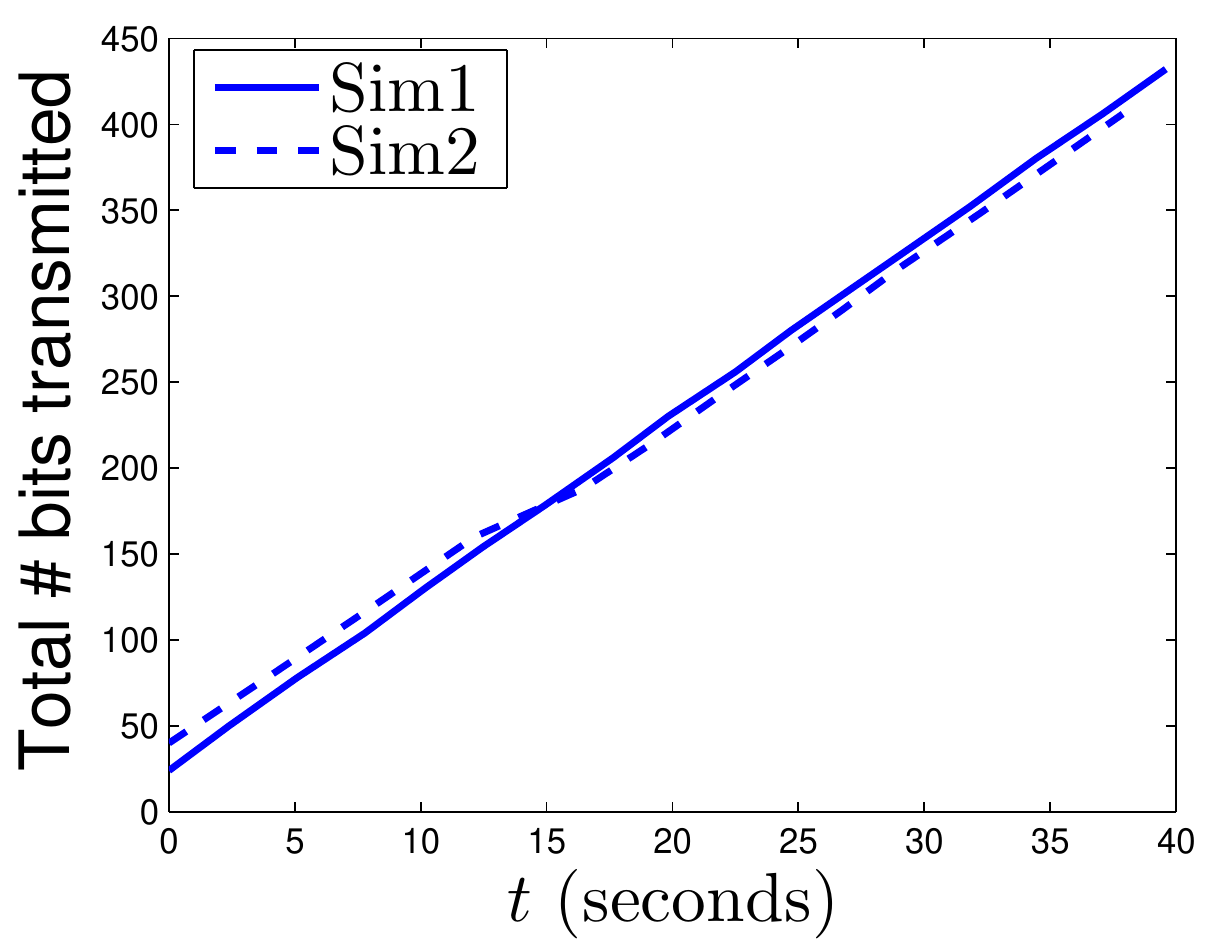}}
  \vspace*{-1ex}
  \caption{For non-instantaneous communication without disturbance and
    $\bar{p} = 20$, (a) shows the number of bits on each transmission
    for ``Sim2'' (b) shows a comparison of the interpolated total
    number of bits transmitted in ``Sim1,2''.}\label{fig:bit_rate_non_inst}
\end{figure}

\vspace*{-2ex}
\section{Conclusions}\label{sec:conc}

We have studied the problem of exponential practical stabilization of
linear-time invariant systems, in the presence of disturbance, and
under bounded communication bit rates.  Our event-triggered design
opportunistically determines the times for communication as well as
the numbers of bits to be transmitted at each time.  Given a uniform
bound on the norm of the disturbance and a prescribed rate of
convergence, the control strategy proposed here asymptotically
confines the plant to a compact set, guarantees a uniform positive
lower bound on inter-transmission and inter-reception communication
times, and ensures that the number of bits transmitted at each
transmission is uniformly upper bounded. These guarantees are valid
for instantaneous transmissions with finite precision data as well as
for non-instantaneous transmissions with bounded communication rate.
The combination of elements from event-triggered control and
information theory has also enabled us to guarantee an arbitrarily
prescribed convergence rate (something not typically ensured in the
information-theoretic approach) and characterize necessary and
sufficient conditions on the number of bits required for stabilization
under opportunistic transmissions (an issue mostly overlooked in
event-triggered control).

Among the limitations of our work are the assumption of known
communication delays and the related requirement of synchronized
updates by the encoder and the decoder for maintaining a synchronized
quantization domain. Additionally, the coding scheme we have used,
though simple, is conservative and introduces a gap between the
necessary and sufficient data rates.  Future work will also explore
better characterization of data rates under disturbances, the
characterization of the gain in performance of dynamic controllers
over static ones, the extension of the results to stochastic
time-varying communication channels, and, more generally, the
understanding of the trade-offs between system performance and
timeliness and size of transmissions.

%% FOR FINAL VERSION
\section*{Acknowledgments}
The authors would like to thank Professor Massimo Franceschetti for
several discussions and the reviewers for various suggestions that
improved the presentation.
This research was supported by NSF Award CNS-1329619.

\appendix

\begin{proof}[Proof of Lemma~\ref{lem:bound-b}]
  The proof is based on direct calculations and on the Comparison
  Lemma~\cite{HKK:02}. We start by noting that $w > 0$ follows
  from~\eqref{eq:W}. From \eqref{eqn:plant_dyn_enc_err}, the Lie
  derivative of $V$ along the flow of the closed-loop dynamics is
  \begin{align}\label{eqn:Vdot}
    \dot{V}(t) &= - x^T(t) Q x(t) - 2 x^T(t) PBK x_e(t) + 2 x^T(t) P
    v(t) \notag
    \\
    &\leq - \frac{\lambda_m(Q)}{\lambda_M(P)} V(x(t)) + 2
    \frac{\sqrt{V(x(t))}}{\sqrt{\lambda_m(P)}} \Enorm{PBK}
    \Enorm{x_e(t)} + \notag
    \\
    & \qquad 2 \frac{\sqrt{V(x(t))}}{\sqrt{\lambda_m(P)}} \Enorm{P}
    \nu
  \end{align}
  where we have used the fact that $P$ satisfies~\eqref{eq:Lyap-eq} as
  well as~\eqref{eq:quadratic-form-ineq}
  and~\eqref{eqn:disturb_bound}. Similarly to the derivation
  of~\eqref{eqn:de_evolve}, we have for $t \in [t_k, t_{k+1})$,
  \begin{align*}
    \Enorm{x_e(t)} &= \Enorm{e^{A(t-t_k)}} \Enorm{x_e(t_k)} +
    \frac{\nu}{\Enorm{A}} [ e^{\Enorm{A} (t-t_k)} - 1 ]
    \\
    & \leq \sqrt{n} e^{\Enorm{A} (t-t_k)} c \sqrt{V_d(t_k)}
    \epsilon(t_k)
    \\
    &\qquad + \frac{\nu}{\Enorm{A}} [ e^{\Enorm{A} (t-t_k)} - 1 ] ,
  \end{align*}
  where we have used $ d_e(t_k) = c \epsilon(t_k)
  \sqrt{V_d(t_k)}$. Substituting this expression in~\eqref{eqn:Vdot},
  we have for $t \in [t_k, t_{k+1})$
  \begin{align*}
    \dot{V}(t) &\leq - \frac{\lambda_m(Q)}{\lambda_M(P)} V(x(t)) + 2
    \frac{\sqrt{V(x(t))}}{\sqrt{\lambda_m(P)}} \Enorm{P} \nu \notag
    \\
    & \qquad + W \sqrt{V(x(t))} e^{\Enorm{A}(t - t_k)} \sqrt{V_d(t_k)}
    \epsilon(t_k)
    \\
    & \qquad + \frac{W}{c} \sqrt{V(x(t))} \frac{\nu}{\Enorm{A}} (
    e^{\Enorm{A} (t-t_k)} - 1 ) .
  \end{align*}
  From the definition~\eqref{eqn:bt_def} of $b$, we compute
  \begin{align*}
    \dot{b} &= \frac{\dot{V} V_d - V \dot{V_d}}{V_d^2} = \frac{
      \dot{V} }{V_d} + \beta b \frac{(V_d - V_0)}{V_d} \leq \frac{
      \dot{V} }{V_d} + \beta b ,
  \end{align*}
  where the inequality follows from the fact that $V_d$ is always
  positive and greater than $V_0$.  Substituting in this equation the
  upper bound for $\dot{V}$ obtained above, we get
  \begin{align*}
    \dot{b} &\leq - w b + \frac{2 \Enorm{P}}{\sqrt{\lambda_m(P)}}
    \frac{\nu \sqrt{b}}{\sqrt{V_d}} + W \epsilon(t_k) e^{\theta \tau}
    \sqrt{b} +
    \\
    &\qquad \frac{W}{c \Enorm{A}} \frac{\nu \sqrt{b}}{\sqrt{V_d}} (
    e^{\Enorm{A} \tau} - 1 ),
  \end{align*}
  where $t = \tau + t_k$. We can further simplify this by noting that
  our region of interest is when the value of $b$ belongs to $[0, 1]$,
  in which $\sqrt{b} \leq 1$, and that $V_d(t) \geq V_0$ for all time
  $t \geq t_0$. Thus,
  \begin{align*}
    \dot{b} &\leq - w b + W \epsilon(t_k) e^{\theta \tau} + c_1 + c_2
    ( e^{\Enorm{A} \tau} - 1 ) .
  \end{align*}
  Thus, letting
  \begin{equation}\label{eqn:btilde_diff_eq}
    \frac{\mathrm{d} \tilde{b}}{\mathrm{d}  \tau} \triangleq - w 
    \tilde{b} + W \epsilon(t_k) e^{\theta \tau} + c_1 + c_2 ( 
    e^{\Enorm{A} \tau} - 1 ) ,
  \end{equation}
  the result follows from the Comparison Lemma.
\end{proof}

\begin{proof}[Proof of Lemma~\ref{lem:Gamma1_prop}]
  To show \textit{(i)}, note that $\tilde{b}(0,1,1)=1$ and
  \begin{equation*}
    \frac{\mathrm{d} \tilde{b}}{\mathrm{d}  \tau} (0, 1, 1) = - w + W 
    + c_1 .
  \end{equation*}
  Using~\eqref{eqn:V0_condition}, we deduce that this value is
  strictly negative, and therefore $\Gamma_1(1, 1) > 0$.
  \textit{(ii)} follows from the fact that $\tilde{b}$ is an
  increasing function of its second and third arguments.  To
  show~\textit{(iii)}, observe that
  \begin{align}\label{eq:tildeb-ineq}
    &\tilde{b}(\tau, b_0, \epsilon_0) - \tilde{b}(\tau, 1, 1) \notag
    \\
    &= e^{-w \tau} \bigg[ ( b_0 - 1) + \frac{W ( \epsilon_0 - 1) }{ w
      + \theta } ( e^{(w + \theta) \tau} - 1 ) \bigg] \notag
    \\
    & \leq e^{-w \tau} \bigg[ ( b_0 - 1) + \frac{ 1 - b_0 }{ e^{
        (w+\theta) T } - 1 } ( e^{(w + \theta) \tau} - 1 ) \bigg].
  \end{align}
  Since $b_0 \leq 1$, we see that for all $\tau \in [ 0 , \min \{
  \Gamma_1(1, 1), T \} ]$, $\tilde{b}(\tau, b_0, \epsilon_0) \leq
  \tilde{b}(\tau, 1, 1) \leq 1$, from which the claim follows.
\end{proof}

\begin{proof}[Proof of Lemma~\ref{lem:consistent_algos}]
  It is sufficient to show that the encoder and the decoder have the
  same signals after running their respective algorithms at $\{ r_k
  \}_{k \in \integerspositive}$. Thus, we will show the equivalence of
  the corresponding steps of the two algorithms. The encoder and
  decoder steps will be prefixed by `E' and `D' respectively. Steps E1
  and D1 are identical initialization of the variable $\delta_0$.
  Step D2 is simply running~\eqref{eqn:x_hat} backwards in time to
  obtain $\hat{x}(t_k^-)$. In D3, $z_E$ is simply the message received
  from the encoder that is encoded in E3. In D4, notice that the terms
  within the parenthesis add up to $d_e(t_k^-)$.  Steps D5 through D7
  are exactly identical to steps E5 through E7, respectively with
  identical data. As a consequence, $\hat{x}(t)$ and $d_e(t)$ values
  at the encoder and decoder are synchronized for all time $t \geq
  t_0$. Further, from Steps 6 of the algorithms it is easy to see that
  $t\mapsto \hat{x}(t)$ evolves according to~\eqref{eqn:x_hat}. It is
  also easy to see that $d_e(t)$ definition
  in~\eqref{eqn:coding_scheme} is consistent with its jump updates in
  the algorithms. It remains to be shown that $\Infnorm{x_e(t)} \leq
  d_e(t)$ for all $t \geq t_0$.
  
  First, observe that as a consequence of the fact that $x(t) = 
  \hat{x}(t) + x_e(t)$, \eqref{eqn:xhat_evolve} and 
  \eqref{eqn:enc_err_dyn} we have that
  \begin{equation*}
    x(t) = e^{\bar{A} (t - t_k)} \hat{x}(t_k^-) + e^{A (t-t_k)} 
    x_e(t_k^-) + \int_{t_k}^t e^{A(t-s)} v(s) \mathrm{d}s.
  \end{equation*}
  Specifically, letting $z_k = \hat{x}(t_k^-)$ as in Step 2 of the
  algorithms, consider the solution $y(.)$ that starts at $z_{D,k}$ at
  $t_k$ and under zero disturbance, i.e.,
  \begin{equation*}
    y(t) = e^{\bar{A}(t-t_k)} z_k + e^{A(t-t_k)} (z_{D,k} - z_k )
  \end{equation*}
  and specifically from Step 6 of the algorithms, we have 
  $\hat{x}(r_k) = y(r_k)$. Further, given that $\Infnorm{ x(t_k^-) - 
    z_{D,k} } \leq \delta_k$, then we have
  \begin{align*}
    x_e(r_k) &= x(r_k) - y(r_k)
    \\
    &= e^{A \Delta_k} (x(t_k^-) - z_{D,k}) + \int_{t_k}^{r_k} 
    e^{A(r_k-s)} v(s) \mathrm{d}s,
  \end{align*}
  which implies that
  \begin{align*}
    \Infnorm{x_e(r_k)} \leq \Infnorm{e^{A \Delta_k}} \delta_k + 
    \frac{\nu}{\Enorm{A}} [ e^{\Enorm{A} \Delta_k } - 1 ] = d_e(r_k)
  \end{align*}
  which is exactly the quantity in Steps E7 and D7. For $t \in [r_k, 
  r_{k+1})$ for $k \in \integersnonnegative$ clearly 
  $\Infnorm{x_e(t)} \leq d_e(t)$, which completes the proof.
\end{proof}

\begin{proof}[Proof of Lemma~\ref{lem:Gamma1-T_sign}]
  Given~\eqref{eqn:Gamma1_def} and considering $b_0$ and $\epsilon_0$
  as parameters, it is sufficient to show that the equation
  $\tilde{b}(\tau, b_0, \epsilon_0) = 1$ has at most one solution in
  the interval $(0,\infty)$. Recall the functions $f_1$ and $f_2$ in
  the definition of $\tilde{b}$ of
  Lemma~\ref{lem:bound-b}. Considering $b_0$ and $\epsilon_0$ as
  parameters, note that the solutions of the equation $\tilde{b}(
  \tau, b_0, \epsilon_0 ) = 1$ are exactly those of $f_1(\tau, b_0,
  \epsilon_0) = f_2(\tau)$, while $\tilde{b}( \tau, b_0, \epsilon_0 )
  < 1$ iff $f_1(\tau, b_0, \epsilon_0) < f_2(\tau)$.
	
  Since $w > 0$, $f_2$ is monotonically increasing. Next, note that
  $\theta = \Enorm{A} + \beta/2 > 0$. Thus, $f_1$ contains the
  dominant exponent and hence there is a $\tau_1(\epsilon_0) \geq 0$
  such that $\dot{f}_1(\tau, b_0, \epsilon_0) > \dot{f}_2(\tau)$ for
  all $\tau > \tau_1(\epsilon_0)$ and $\dot{f}_1(\tau, b_0,
  \epsilon_0) < \dot{f}_2(\tau)$ for all $\tau < \tau_1(\epsilon_0)$.
  Thus, for each $b_0 \in [0,1)$, there exists a unique solution for
  $\tilde{b}(\tau, b_0, \epsilon_0) = 1$. For $b_0 = 1$ and
  $\tau_1(\epsilon_0) > 0$ there exists a unique solution to the
  problem. For $b_0 = 1$ and $\tau_1(\epsilon_0) \leq 0$ there exists
  no solution and $f_1(\tau, b_0, \epsilon_0) > f_2(\tau)$ for all
  $\tau > 0$. In each scenario the claim of the lemma follows
  directly.
\end{proof}

\begin{proof}[Proof of Lemma~\ref{lem:Gamma2-T_sign}]
  Considering $b_0$ and $\epsilon_0$ as parameters, it is sufficient
  to show that $\triggerChbar(\tau, b_0, \epsilon_0, \psi_0) = 1$ has
  a unique solution. We show the uniqueness through a contradiction
  argument. Suppose there exists a $\tau^* > \tilde{\Gamma}_2(b_0,
  \epsilon_0, \psi_0)$ such that $\triggerChbar(\tau^*, b_0,
  \epsilon_0, \psi_0) = 1$. Since $\triggerChbar$ is a continuous
  function, it must then have a local maximum in the time interval
  $[\tilde{\Gamma}_2(b_0, \epsilon_0), \tau^*]$. Notice
  from~\eqref{eqn:hbar_def} that the numerator of $\triggerChbar$ is a
  monotonously increasing function of time $\tau$. Next, since
  $\tilde{b} \mapsto \rhofun{T}{ \tilde{b} }$ is a decreasing function
  it follows that $\tilde{b}(., b_0, \epsilon_0)$ must have a local
  maximum in the time interval $[\tilde{\Gamma}_2(b_0, \epsilon_0),
  \tau^*]$. Thus, considering $b_0$ and $\epsilon_0$ as parameters,
  notice that
  \begin{equation*}
    \frac{\mathrm{d} \tilde{b}}{\mathrm{d} \tau} = - w \tilde{b} + W 
    \epsilon_0 e^{\theta \tau} + c_1 + c_2 ( e^{\Enorm{A} \tau} - 1),
  \end{equation*}
  while the second derivative is
  \begin{equation*}
    \frac{\mathrm{d}^2 \tilde{b}}{\mathrm{d} \tau^2} = - w 
    \frac{\mathrm{d} \tilde{b}}{\mathrm{d} \tau} + W \epsilon_0 
    \theta e^{\theta \tau} + c_2 \Enorm{A} e^{\Enorm{A} \tau}.
  \end{equation*}
  Then notice that the second derivative at any critical point of
  $\tilde{b}$ is positive since the first term vanishes at a critical
  point of $\tilde{b}$, the second term is positive for any $\tau$
  because $\theta > 0$ and $c_2 \Enorm{A} \geq 0$ by definition. Thus
  $\tilde{b}$ as a function of $\tau$ has no local maximum. Thus, this
  contradiction proves the result.
\end{proof}

% \bibliographystyle{IEEEtran}
% \bibliography{alias,FB,JC,Main,Main-add}

\begin{IEEEbiography}[{\includegraphics[width=1in,
    height=1.25in,clip,keepaspectratio]
    {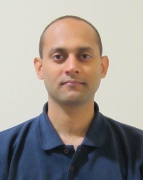}}]{Pavankumar Tallapragada}
  received the B.E. degree in Instrumentation Engineering from SGGS
  Institute of Engineering $\&$ Technology, Nanded, India in 2005,
  M.Sc. (Engg.) degree in Instrumentation from the Indian Institute of
  Science, Bangalore, India in 2007 and the Ph.D. degree in Mechanical
  Engineering from the University of Maryland, College Park in
  2013. He is currently a Postdoctoral Scholar in the Department of
  Mechanical and Aerospace Engineering at the University of
  California, San Diego. His research interests include
  event-triggered control, networked control systems, distributed
  control and transportation and traffic systems.
\end{IEEEbiography}

\begin{IEEEbiography}[{\includegraphics[width=1in,
  height=1.25in,clip,keepaspectratio]{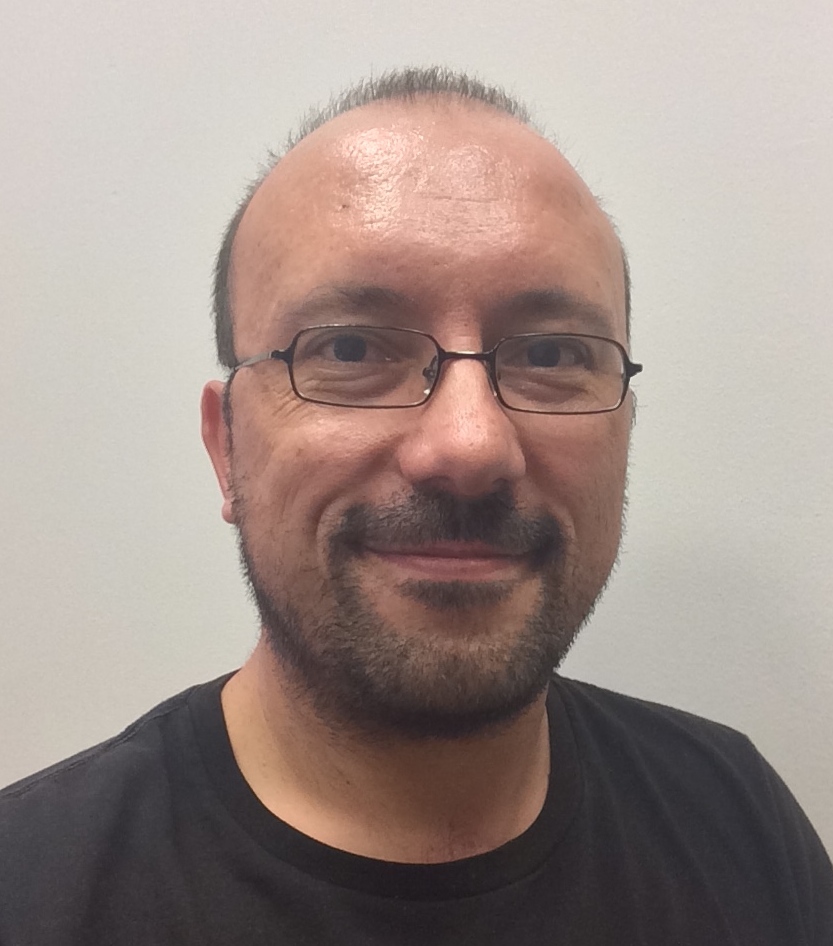}}]{Jorge
    Cort\'es}
  received the Licenciatura degree in mathematics from Universidad de
  Zaragoza, Zaragoza, Spain, in 1997, and the Ph.D. degree in
  engineering mathematics from Universidad Carlos III de Madrid,
  Madrid, Spain, in 2001.  He held post-doctoral positions with the
  University of Twente, Twente, The Netherlands, and the University of
  Illinois at Urbana-Champaign, Urbana, IL, USA. He was an Assistant
  Professor with the Department of Applied Mathematics and Statistics,
  University of California, Santa Cruz, CA, USA, from 2004 to 2007. He
  is currently a Professor in the Department of Mechanical and
  Aerospace Engineering, University of California, San Diego, CA,
  USA. He is the author of Geometric, Control and Numerical Aspects of
  Nonholonomic Systems (Springer-Verlag, 2002) and co-author (together
  with F. Bullo and S. Mart\'inez) of Distributed Control of Robotic
  Networks (Princeton University Press, 2009). He is an IEEE Fellow
  and an IEEE Control Systems Society Distinguished Lecturer.  His
  current research interests include distributed control, networked
  games, power networks, distributed optimization, spatial estimation,
  and geometric mechanics.
\end{IEEEbiography}

\end{document}